\documentclass[aps,pra,twocolumn,10pt,superscriptaddress,notitlepage,floatfix]{revtex4-1}

\usepackage[dvips]{graphicx} 
\usepackage{amsmath,amssymb,amsthm,mathrsfs,amsfonts,dsfont}
\usepackage{enumerate}
\usepackage{colonequals}
\usepackage{bbm}
\usepackage{epsfig}
\usepackage{subfigure}
\usepackage{xcolor}
\usepackage{braket}
\usepackage{multirow}    
\usepackage{array}     
\usepackage{comment}
\usepackage{enumitem} 
\usepackage{tcolorbox}
\usepackage{framed}
\usepackage{setspace}
\usepackage{fancyhdr}
\usepackage{extramarks}
\usepackage{chngpage}
\usepackage{float,wrapfig}
\usepackage{CJK}
\usepackage{ifthen}
\usepackage{tikz}
\usepackage{qcircuit}
\usepackage{hyperref}
\hypersetup{
	colorlinks,
	linkcolor={blue},
	citecolor={blue},
	urlcolor={blue},
	pdftitle={Fast Estimation of Sparse Quantum Noise},
	pdfauthor={Robin Harper, Wenjun Yu, Steven T Flammia}
}
\usepackage{cleveref}
\usepackage{algorithm}
\usepackage{algpseudocode}
\crefformat{section}{\S#2#1#3} 
\crefformat{subsection}{\S#2#1#3}
\crefformat{subsubsection}{\S#2#1#3}
\crefrangeformat{section}{\S\S#3#1#4 to~#5#2#6}
\crefmultiformat{section}{\S\S#2#1#3}{ and~#2#1#3}{, #2#1#3}{ and~#2#1#3}
\newcommand{\Prob}[1]{\mathrm{Pr}\left({#1}\right)}
\newcommand{\Var}[1]{\mathrm{Var}\left({#1}\right)}
\newcommand{\Cov}[1]{\mathrm{Cov}\left({#1}\right)}
\allowdisplaybreaks

\AtBeginDocument{
\heavyrulewidth=.08em
\lightrulewidth=.05em
\cmidrulewidth=.03em
\belowrulesep=.65ex
\belowbottomsep=0pt
\aboverulesep=.4ex
\abovetopsep=0pt
\cmidrulesep=\doublerulesep
\cmidrulekern=.5em
\defaultaddspace=.5em
\tabcolsep=7pt
}
\usepackage{booktabs}
\makeatletter
\newcommand\fs@booktabsruled{%
  \def\@fs@cfont{\bfseries\strut}\let\@fs@capt\floatc@ruled
  \def\@fs@pre{\hrule height\heavyrulewidth depth0pt \kern\belowrulesep}%
  \def\@fs@mid{\kern\aboverulesep\hrule height\lightrulewidth\kern\belowrulesep}%
  \def\@fs@post{\kern\aboverulesep\hrule height\heavyrulewidth\relax}%
  \let\@fs@iftopcapt\iftrue
}
\makeatother
\floatstyle{booktabsruled}
\restylefloat{algorithm}

\usepackage{etoolbox}
\AtBeginEnvironment{algorithm}{\linespread{1.4}\selectfont}

\newtheorem{lemma}{Lemma}
\newtheorem{definition}{Definition}
\newtheorem{remark}{Remark}
\newtheorem{assumption}{Assumptions}
\newtheorem{theorem}{Theorem}
\newenvironment{thmbis}[1]
{\renewcommand{\thetheorem}{\ref{#1}}%
   \addtocounter{theorem}{-1}%
   \begin{theorem}}
  {\end{theorem}}

\newcommand{\sgn}[1]{\mathsf{sgn}\left[{#1}\right]}

\begin{document}

\title{Fast Estimation of Sparse Quantum Noise}

\author{Robin Harper}
\thanks{These authors contributed equally.}
\affiliation{Centre for Engineered Quantum Systems, School of Physics, University of Sydney, Sydney, NSW 2006 Australia}
\author{Wenjun Yu}
\thanks{These authors contributed equally.}
\affiliation{Institute for Interdisciplinary Information Sciences, Tsinghua University, Beijing 100084, China}
\author{Steven T.\ Flammia}
\affiliation{AWS Center for Quantum Computing, Pasadena, CA 91125 USA}

\begin{abstract}
As quantum computers approach the fault tolerance threshold, diagnosing and characterizing the noise on large scale quantum devices is increasingly important. 
One of the most important classes of noise channels is the class of Pauli channels, for reasons of both theoretical tractability and experimental relevance.  
Here we present a practical algorithm for estimating the $s$ nonzero Pauli error rates in an $s$-sparse, $n$-qubit Pauli noise channel, or more generally the $s$ largest Pauli error rates.
The algorithm comes with rigorous recovery guarantees and uses only $O(n^2)$ measurements, $O(s n^2)$ classical processing time, and Clifford quantum circuits. 
We experimentally validate a heuristic version of the algorithm that uses simplified Clifford circuits on data from an IBM 14-qubit superconducting device and our open source implementation.
These data show that accurate and precise estimation of the probability of arbitrary-weight Pauli errors is possible even when the signal is two orders of magnitude below the measurement noise floor.
\end{abstract}

\date{\today}

\maketitle

\section{Introduction}

Estimating noise in quantum computers is becoming increasingly important as we begin to test quantum error correction (QEC) on current noisy intermediate-scale devices~\cite{Martinis2015}. 
Much of the current effort in noise estimation is focused on identifying methods that will remain tractable as the system size increases beyond the few qubit regime~\cite{Cramer2010a, Lanyon2017, Helsen2018, Proctor2018, Erhard2019, Sarovar2019, Bairey_2019, Bairey_2020, dumitrescu2019hamiltonian, evans2019scalable, huang2020predicting, hamilton2020scalable, torlai2020quantum, klimov2020snake}. 
In such larger systems it is important to identify not only the errors that occur when qubits are operated in isolation or in small groups but also the additional errors that occur when the device is implementing fault-tolerant QEC circuits and nontrivial quantum algorithms.
If we are able to characterize the noise and noise types (such as control errors, decoherence and crosstalk errors) in such a system then that will allow us to better diagnose and fix such errors, for instance by enabling calibration in the presence of crosstalk. 
Characterization of the noise will also allow the construction of tailored quantum error-correcting codes and decoders and customized fault-tolerance protocols designed to counteract the specific noise in the system. 
Such bespoke systems have been shown to outperform their generic counterparts at quantum error correction~\cite{Aliferis2008, Tuckett2018, Puri2019, Guillaud2019, Tuckett2019, Tuckett2020}.

Noise estimation is possible in principle using quantum process tomography \cite{Chuang1997}, but in practice this is often not desirable for several reasons. 
First, even using methods such as compressed sensing \cite{Gross2010, Shabani2011, Flammia2012, Rodionov2014, Kalev2015, Riofro2017}, the enormous Hilbert space of a multi-qubit machine makes it difficult to efficiently estimate all possible parameters beyond a handful of qubits. 
Second, standard tomography protocols are susceptible to state preparation and measurement (SPAM) errors~\cite{Merkel2012}, which limit the accuracy in estimating noise in quantum gates.

One promising approach to make noise characterization more tractable is to reduce the noise to a smaller set of relevant parameters that can be estimated in a SPAM-free way. 
A natural candidate for this approach is to learn the \textit{Pauli projection} of a quantum noise channel.
This is the channel obtained when the noise channel is twirled over the set of $n$-qubit Pauli operators.
The remaining parameters of the channel, known as the Pauli error rates, are the most relevant parameters for near-term applications of QEC and fault tolerance because of the dominant role played by stabilizer codes~\cite{Terhal2015}.
Moreover, practical methodologies have been developed to implement the Pauli projection without substantially changing the average error rate in a given round of gates \cite{Knill2005, Wallman2016, Ware2018}. 
Furthermore, QEC tends to make noise less coherent~\cite{Huang2018,Beale2018,iverson2019coherence}, which further justifies the Pauli approximation at the logical level. 
Finally, Pauli error rates can be learned in a SPAM-free way~\cite{Flammia2019,Harper2019}. 

Focusing on Pauli channels reduces the number of parameters required for complete noise estimation to $4^n$, where $n$ is the number of qubits of the device. 
Although this has better scaling than other SPAM-robust methods that attempt to learn an entire noise channel (e.g.~\cite{Kimmel2013, Blume-Kohout2016}), this is unfortunately already too large to be tractable for some present-day quantum devices~\cite{Arute2019}. 
There are several ways to try to reduce this parameter count even further while still capturing the most relevant parameters for fault tolerance and QEC. 
For example, when the Pauli error rates form a bounded-degree Markov field, then the channel can be learned efficiently in $n$~\cite{Flammia2019}; this algorithm was experimentally validated in Ref.~\cite{Harper2019}. 
Ref.~\cite{Flammia2019} also gave an efficient algorithm for estimating the class of $s$-sparse Pauli channels, i.e.\ those with at most $s$ nonzero Pauli error rates.
These two classes of Pauli channels are motivated by the fact that quantum devices approaching the fault-tolerant regime will have very few significant errors (and therefore are approximately sparse) and will have errors that are only weakly correlated (and therefore are approximated by a low-degree Markov field).

\subsection{Main Results}

In this paper, we give a new algorithm for estimating $s$-sparse Pauli channels that is distinct from Ref.~\cite{Flammia2019}. 
This algorithm can reconstruct an $s$-sparse Pauli channel with the following recovery guarantee. 
We assume first that an experiment can be modeled as having access to a noisy oracle that can return an eigenvalue of an unknown Pauli channel with some independent Gaussian noise with variance $\xi^2$. 
Then using at most $O(sn)$ queries to the noisy oracle, the algorithm returns $s$ estimated error rates $\hat{p}_j$ that agree with the channel error rates $p_j$ with precision $|\hat{p}_j - p_j| \le O\bigl(\frac{\xi}{\sqrt{s}}\bigr)$. 
In fact, the bound is slightly stronger than this.
The precise statement is given in \Cref{Thm:main}, together with \Cref{random_support_assumption} which lay out the precise mathematical assumptions used in the derivation.

We then show how to break open the oracle and perform the entire estimation efficiently.  
We show that noisy eigenvalues can be estimated to within variance $\xi^2$ by using only Clifford quantum circuits and computational basis measurements.
Our results use modifications of the algorithm from~\cite{Flammia2019} and show how the relevant noisy eigenvalue queries can be obtained with only  $O\bigl(\frac{n^2}{\xi^2}\bigr)$ measurements. 

Next, we validate these algorithms using experimental data from a 14-qubit superconducting device~\cite{Harper2019}. 
The original experiment exhaustively estimated the averaged eigenvalues in this device. 
We use these data to construct our eigenvalue oracle. 
We then simulate various levels of measurement noise on top of this ``true'' experimental signal to validate our algorithms. 
Our results are depicted in \Cref{fig:initial_results}. 
We show that when the noise added to the eigenvalues has any standard deviation in the range of $10^{-3}$--$10^{-5}$ then we can accurately recover Pauli error rates as small as \textit{two orders of magnitude} less than the noise added on the eigenvalues. 
Importantly, even when we artificially add arbitrary many-body Pauli errors with comparable error probabilities, we still recover these strongly correlated errors with high relative precision. 

Our results suggest that practical characterization of all Pauli error rates with probabilities greater than $10^{-4}$ or $10^{-5}$ in a quantum device with $10$--$20$ qubits can be achieved with around $10^6$ or $10^7$ experimental measurements. 
In such quantum devices having sub-microsecond gate times, this puts practical noise characterization within reach on a time scale of hours, not days or weeks. 

Finally, we have written open source code, available on GitHub~\cite{GITHUB}, which reproduces all the figures in this manuscript and contains other examples which explain how to use the algorithms in real experiments.

The remainder of this paper is organized as follows. 
We provide some notation and background in \cref{sec:prelim} followed by an intuitive overview of our recovery algorithm in \cref{sec:sparse}. 
We state our precise recovery guarantees in \cref{sec:guarantee}.  
We describe the circuits we use for practical eigenvalue estimation and provide details of our validation results in \cref{sec:heuristic,sec:validation}. 
We have deferred the precise definition of the algorithm until \cref{sec:algorithm} and the proofs until \cref{sec:proof,sec:tailbounds}.
We conclude in \cref{sec:conclusion}.

\begin{figure*}[th]
    \centering
    \begin{tikzpicture}
	\node[inner sep=0pt] at (2.5,4.5)         {\includegraphics[width=1.6\columnwidth]{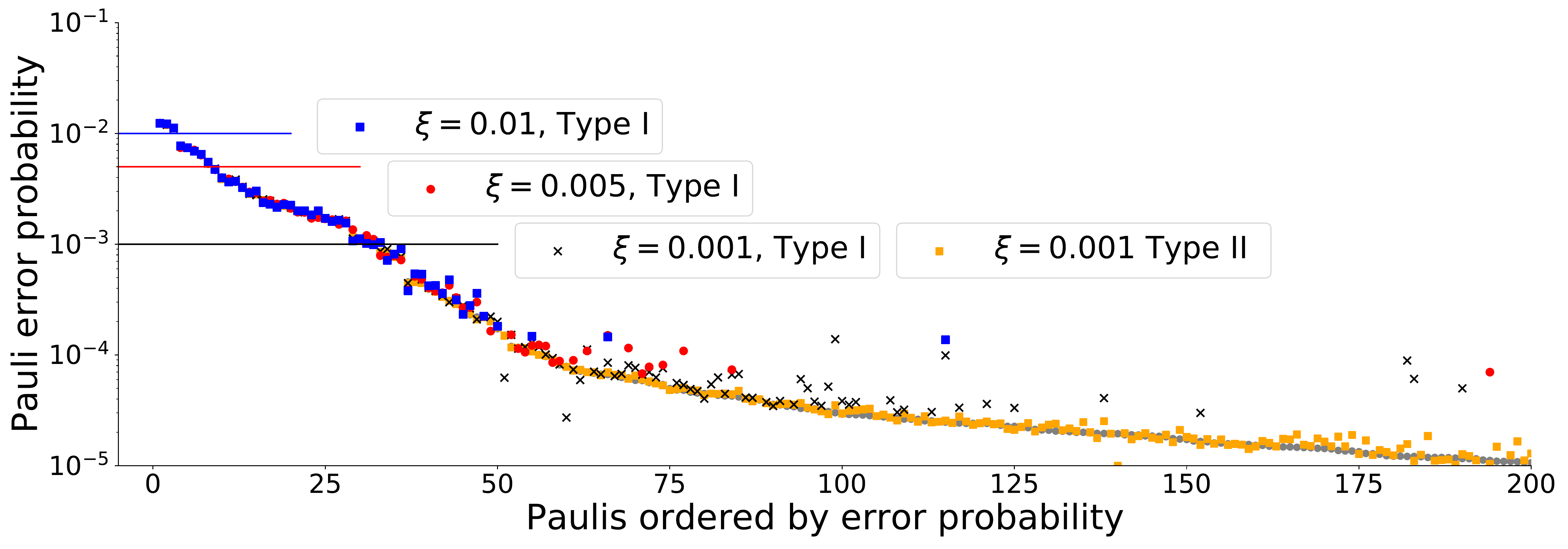}};
	\node at (-4.5,7) {\textbf{(a)}};
	\node[inner sep=0pt] at (2.5,-4)         {\includegraphics[width=1.8\columnwidth]{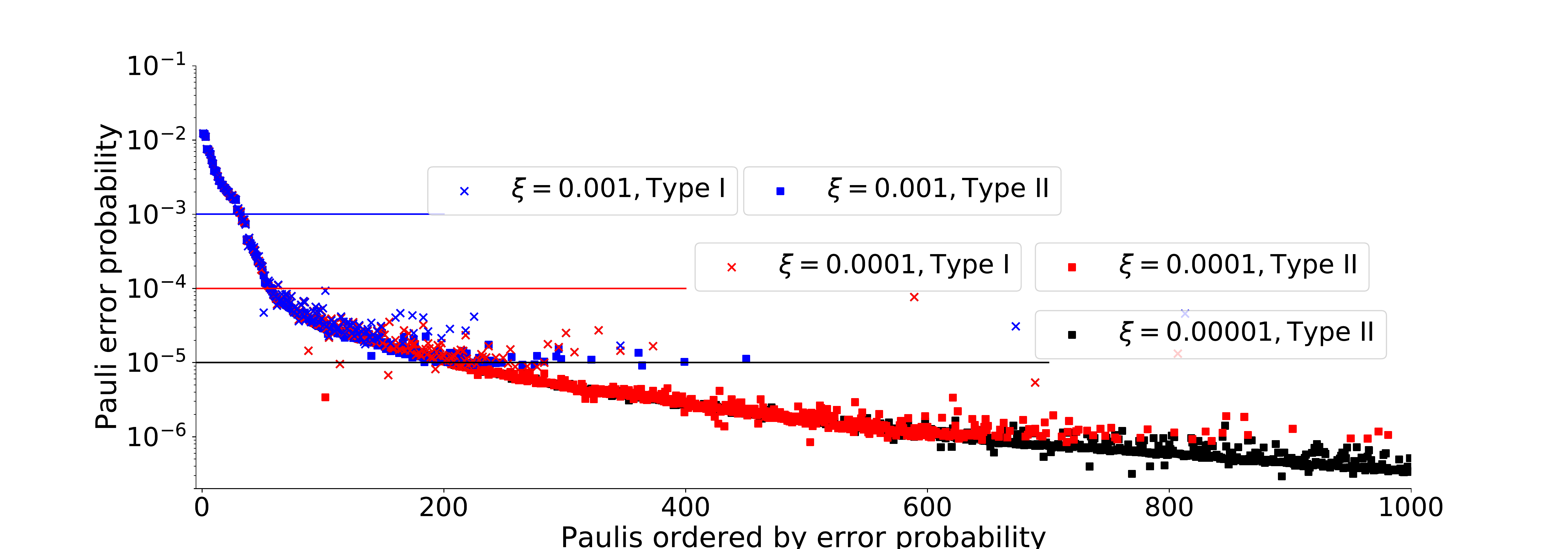}};
	\node[inner sep=0pt] at (10,5)         {\includegraphics[width=0.62\columnwidth]{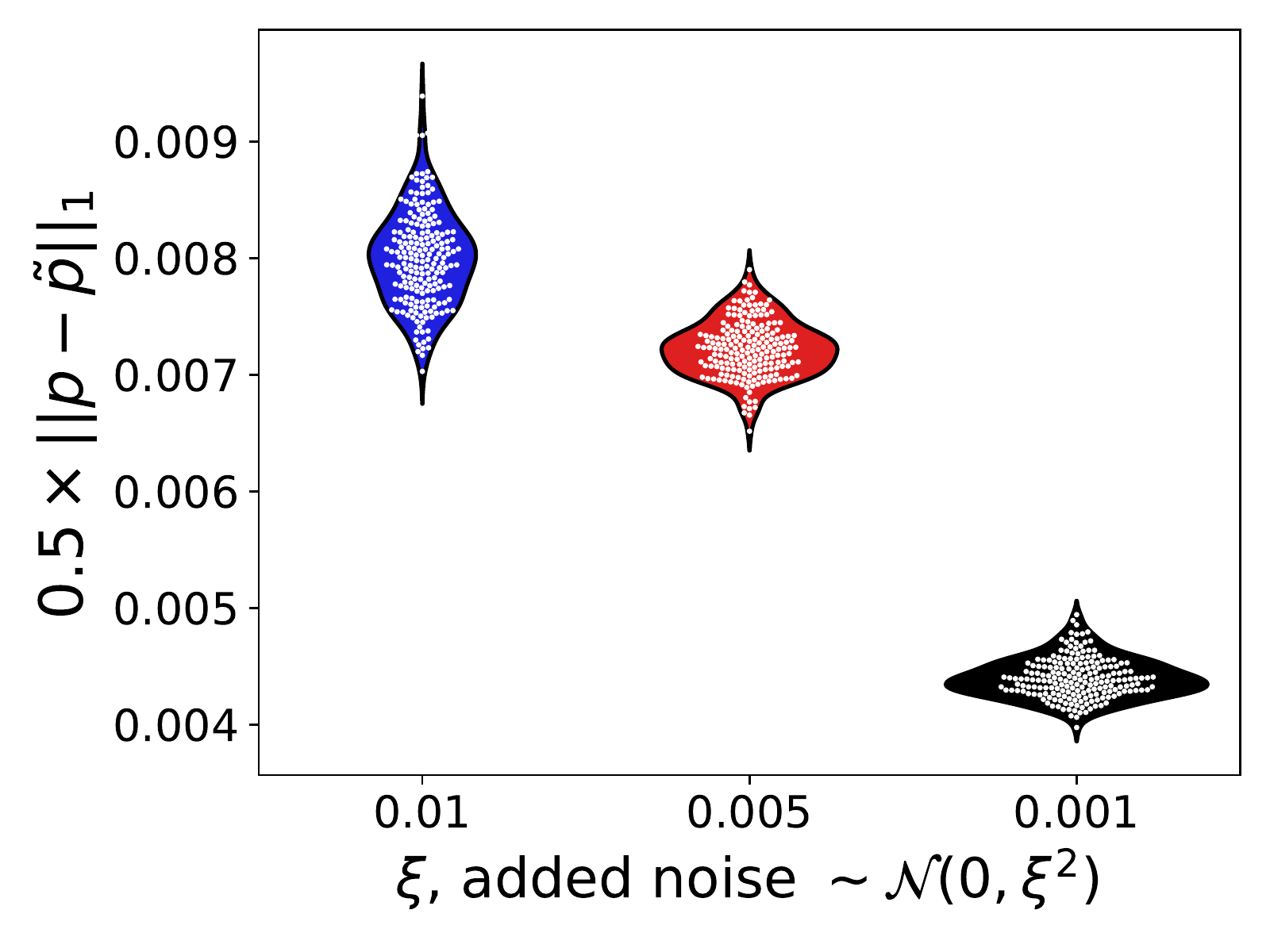}};
	\node at (-4,1.6) {\textbf{(c)}};
	\node (R) at (8.7,-5.6) {};
	\node (B) at (6,-1.2) {};

	\node[inner sep=0pt] at (1.4,0)  {\includegraphics[width=1.1\columnwidth]{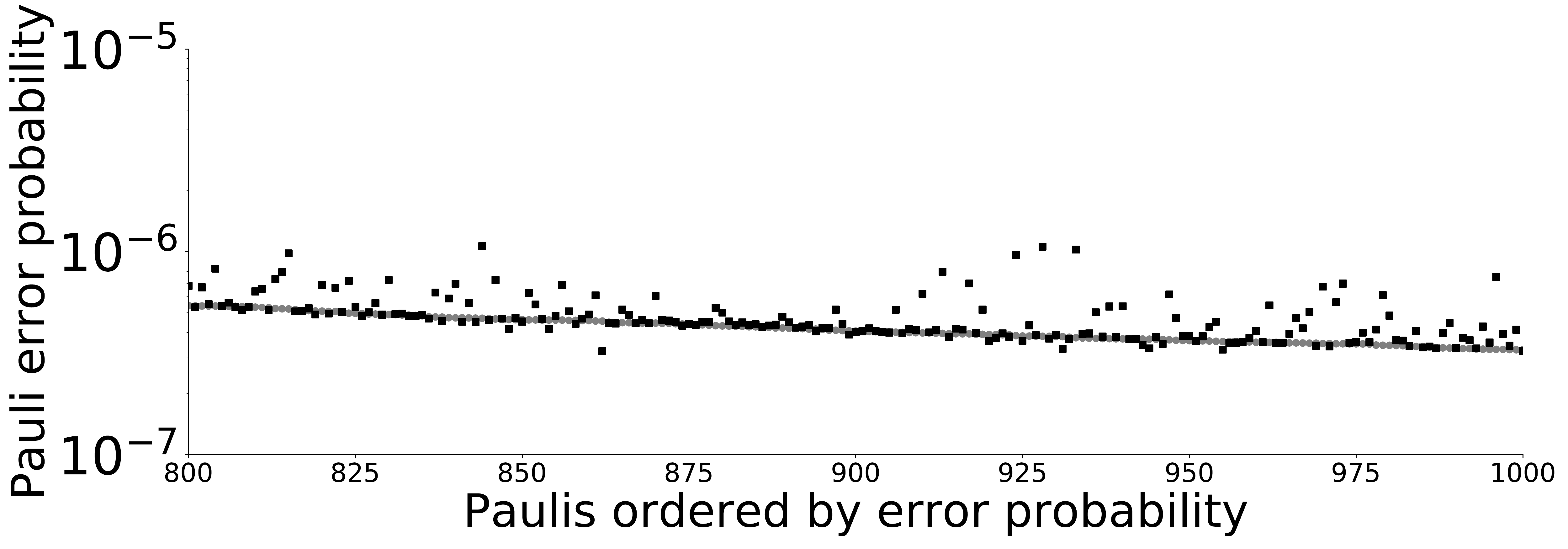}};
	\draw (6.3,-5.6) -- (8.7,-5.6) -- (8.7,-6.1) -- (6.3,-6.1) -- (6.3,-5.6);
	\node at (-4.5,-1.3) {\textbf{(b)}};
	\node at (6.8,7) {\textbf{(d)}};
	\node at (7,1.5) {\textbf{(e)}};
	\node[inner sep=0pt] at (10,-0.4)         {\includegraphics[width=0.6\columnwidth]{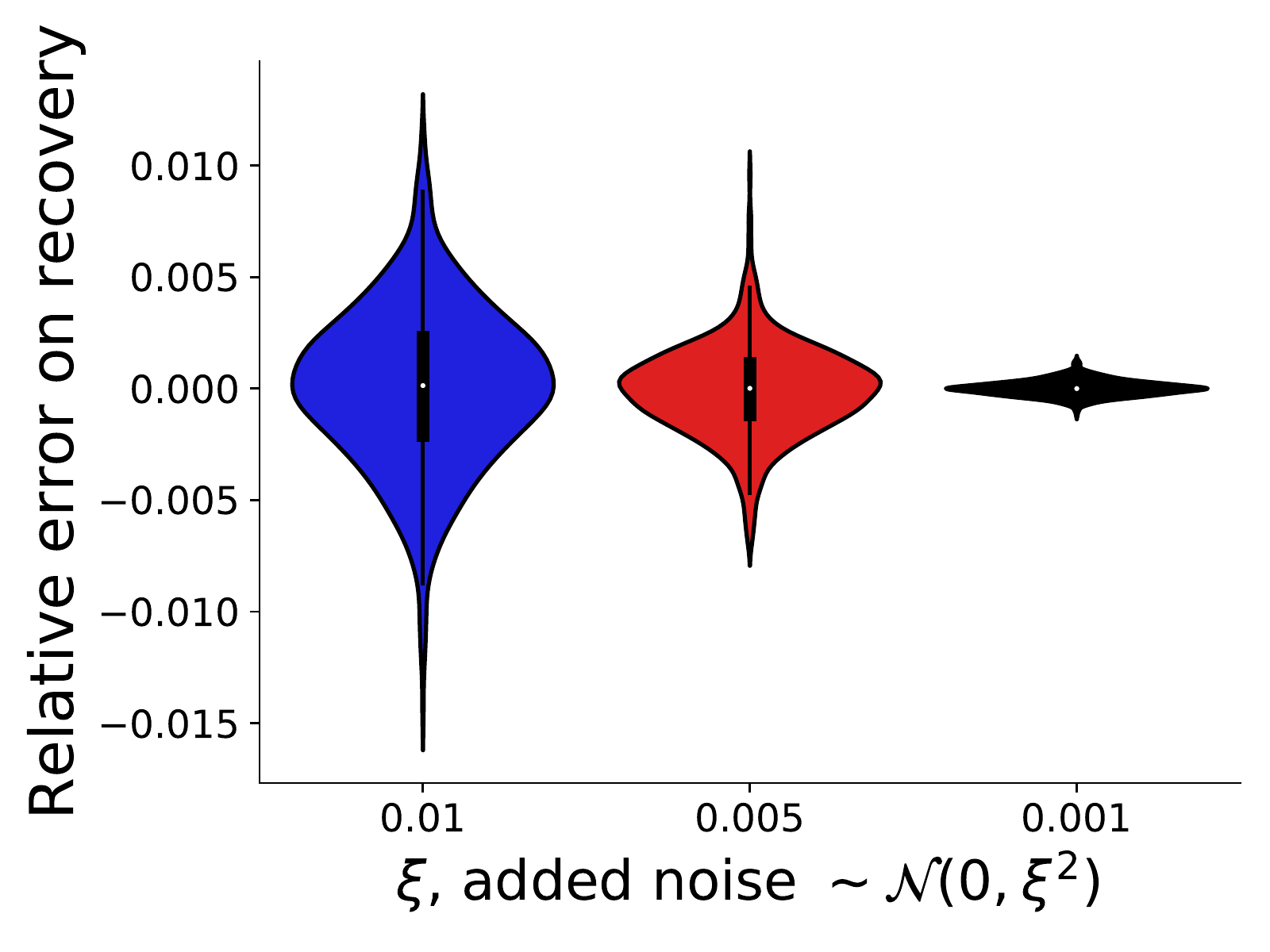}};
	\draw[<->]   (R) to[out=0,in=0] (B);
    
	\end{tikzpicture}	

    \caption{\textbf{(a)} This figure shows the ability of the reconstruction algorithms to recover sparse Pauli error rates from experimental data. 
    The ``true'' error rates (gray line) were constructed using data from a 14-qubit experiment~\cite{Harper2019} as described in the main text (\cref{sec:validation}).
    Dots indicate recovered Pauli error rates using our algorithm with artificially added normally distributed noise on top of the true error rates to simulate finite sampling and other noise sources. 
    The reconstruction used two experimental designs using a number of randomized benchmarking style experiments: Type I used $58$ and Type II used $365$ such experiments, each with the same number of samples per experiment. 
    The Type II experimental design runs more experiments and therefore takes more data overall, but allows recovery of an increasing number of Paulis while keeping constant the number of measurements per experiment. 
    \textbf{(b)} shows the recovery of 1000 different error rates as low as $10^{-7}$ with high relative precision, when the experimental noise varies between $\xi=10^{-3}$--$10^{-5}$. 
    Notably, the error rates are recovered with a precision almost two orders of magnitude below the standard deviation $\xi$ of the noise added to the signal. 
    \textbf{(c)} a more detailed look at the recovery in the regime between $10^{-6}$ and $10^{-7}$ with noise levels of $\xi=10^{-5}$.
    \textbf{(d)} Violin plot of the total variational (1-norm) distance between the original probability distribution ($\boldsymbol{p}$) and the reconstructed probability distribution ($\boldsymbol{\tilde{p}}$), being  $\frac12\|\boldsymbol{p}-\boldsymbol{\hat{p}}\|_1$. 
    The charts show the spread of recovery error over 200 different randomly generated samples of noise. 
    As can be seen the entire probability distributions are consistently recovered to high precision. 
    \textbf{(e)} shows a separate experiment where 4 distinct uniformly random many-body Paulis were added to the oracle with error rates chosen randomly from a normal distribution $\mathcal{N}(0.005, 0.001)$. 
    The algorithm was run with this additional signal to test if it can recover these Paulis as well. 
    In all cases the planted Paulis were recovered with small relative error, as shown.}
    \label{fig:initial_results}
\end{figure*}
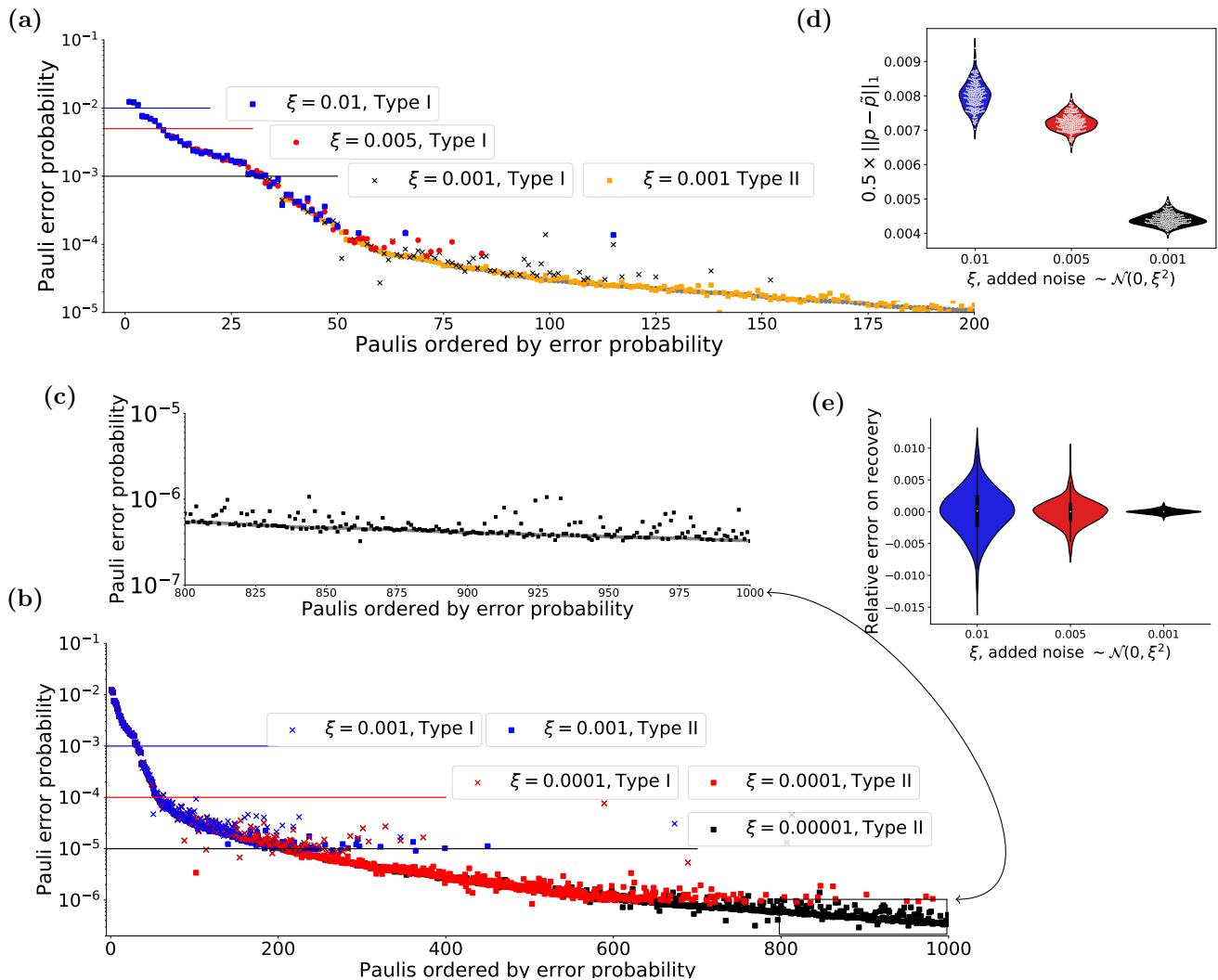

\section{Notation and Background}\label{sec:prelim}

Given a set of $n$ qubits with Hilbert space dimension $2^n$, we can introduce the following notation. 
Let $\boldsymbol{\mathcal{P}}^{n}$ denote the group of Pauli operators on all $n$ qubits and $\boldsymbol{\mathsf{P}}^{n}=\boldsymbol{\mathcal{P}}^n/\langle i\rangle$ be the Paulis modulo phase. 
There is a natural isomorphism between multiplication on $\boldsymbol{\mathsf{P}}^n$ and bit-wise addition of $2n$-bit strings $\mathbb{F}_{2}^{2n}$ given by
\begin{gather}\label{eq:iso}
    a\in\mathbb{F}^{2n}_2,\ a\longleftrightarrow P_a=P_{a_xa_z}=i^{a_x\cdot a_z}X[a_x]Z[a_z],
\end{gather}
where $a_x,a_z\in\mathbb{F}^n_2$ and $X$ and $Z$ are the standard single-qubit Pauli matrices, and $P_a \in \boldsymbol{\mathcal{P}}^n$ is understood to be a canonical coset representative. 
Here $X[a_x] = X^{a_{x_1}}\otimes \ldots \otimes X^{a_{x_n}}$, and similar for $Z[a_z]$.
Using this isomorphism, we can directly use $a\in\mathbb{F}_2^{2n}$ to denote the Pauli matrix $P_a$. 
For any two Pauli matrices $P_a$ and $P_b$, we have $P_aP_b=(-1)^{\langle a,b\rangle}P_bP_a$ where the symplectic inner product
\begin{gather}\label{eq:PauliInnerProduct}
    \langle a,b\rangle=a_x\cdot b_z+a_z\cdot b_x\ \bmod2 
\end{gather}
is symmetric and bilinear.

We define a \textit{stabilizer group} $\boldsymbol{\mathsf{S}}$ to be a linear subspace of $\mathbb{F}^{2n}_2$ such that for all $a,b\in\boldsymbol{\mathsf{S}}$, $\langle a,b\rangle=0$. 
Thus a stabilizer group forms a commuting subgroup of the full Pauli group by the mapping in \eqref{eq:iso}. 

An $n$-qubit \textit{Pauli channel} $\mathcal{E}$ acting on a quantum state $\rho$ is of the form 
\begin{align}\label{eq:errorratedef}
    \mathcal{E}(\rho) = \sum_j p_j P_j \rho P_j\,, 
\end{align}
where $p_j$ is the error rate associated with the Pauli operator $P_j$. 
The \textit{Pauli error rates} $p_j$ form a probability distribution over all $N = 4^n$ elements of the $n$-qubit Pauli group modulo phases. 
These are closely related to, but distinct from, the \textit{Pauli channel eigenvalues}, which are defined as
\begin{align}\label{eq:channeleigdef}
    \lambda_j = \frac{1}{2^n}\mathrm{Tr}\bigl(P_j \mathcal{E}(P_j)\bigr)\,.
\end{align}
Because it will be clear from context, we will often refer to these simply as the ``error rates'' and the ``eigenvalues''.
Thus, when a state $\rho$ is subjected to the noisy channel $\mathcal{E}$, the error rate $p_j$ describes the probability of a multi-qubit Pauli error $P_j$ affecting the system. 
In contrast, the eigenvalues describe how faithfully a given multi-spin Pauli operator is transmitted through the channel.
The error rates $p_j$ and eigenvalues $\lambda_j$ are related by a \textit{Walsh-Hadamard transform} (WHT). 
From equations~\eqref{eq:errorratedef} and \eqref{eq:channeleigdef} and the orthogonality relations of the Pauli group, we can compute the Walsh-Hadamard transform coefficients:
\begin{equation}\label{eq:WHT}
    \lambda_k = \sum_{j\in\mathbb{F}^{2n}_2}(-1)^{\langle k,j\rangle}p_j\,.
\end{equation}
The symmetrical nature of the Walsh-Hadamard transform means we also have the inverse relation:
\begin{equation}\label{eq:reverse:WHT}
    p_j = \frac{1}{N}\sum_{k\in\mathbb{F}^{2n}_2}(-1)^{\langle j,k\rangle}\lambda_k\,.
\end{equation}
Note that our WHT is ordered by Pauli commutation relations---see \Cref{app:WHTordering} for a further discussion of this subtlety. 
Finally, for any natural number $N$, we then write $[N]$ to mean $\{0,\dots,N-1\}$. 

In an analogy with discrete Fourier transforms, the error rates can be thought of as the frequency domain components of the time domain signal, which in this case is the eigenvalues.
Our goal is to sparsely sample the dense time domain signal (the eigenvalues) and reconstruct the entire (but sparse) frequency domain (the error rates).
The theory of compressed sensing allows us to do this in principle with very few measurements, namely $O(s \log N)$. 
However the standard reconstruction methods that use convex optimization require $\textrm{poly}(N)$ classical computation, which is too expensive since $N=4^n$. 
Therefore, unlike in compressed sensing, we must reconstruct the sparse frequency domain signal using only $\textrm{poly}(s,\log N)$ resources for our algorithm to be considered efficient, which is indeed what we achieve here. 

Throughout this paper, we will restrict to a sparsity regime with only $s\ll4^{n/2}$ nonzero error rates, each having probabilities greater than a specified cutoff $\epsilon_0$.
(In our proofs, we assume that any error rate less than $\epsilon_0$ is identically zero, although the heuristic algorithm is more forgiving.) 
This is what we mean when we refer to an $s$-sparse model. 
This allows our algorithm to perform in the regime where $s$ is exponential in $n$. 
When such an exponential scaling holds, it makes our algorithm inefficient in $n$, but this is also a relevant regime if we wish to estimate Pauli channels with an extensive entropy. 
Distributions with extensive entropy will generally require an exponential number of error rates to estimate them with arbitrary accuracy. 

Our recovery methodology builds on one of the main results of Ref.~\cite{Flammia2019} and an adaptation of the classical algorithms described in Refs.~\cite{Scheibler2015,Li2015}. 
In Ref.~\cite{Flammia2019}, the authors show how to recover all $N = 4^n$ Pauli channel eigenvalues to \textit{relative} precision $\epsilon$ using $O\bigl(\epsilon^{-2}n 2^n\bigr)$ measurements. 
The circuit modifications we require on top of that algorithm are shown in \Cref{fig:circuits}.
The recovery of all $N$ eigenvalues would require $2^n+1$ applications of depth $O(m+ n^2/\log n)$ Clifford circuits, or $3^n$ applications of depth $O(m)$ Clifford circuits. 
Here $m$ is a constant that depends on the channel being estimated and is $O(1/\Delta)$ with $\Delta = 1-\max_{j\neq 0}\lambda_{j}$ being the spectral gap of the channel. 
We can consider $m = O(1)$, although the implied constant might be large for high-fidelity quantum channels.
While the depth of this algorithm is efficient, the number of distinct circuits required is clearly not scalable in $n$. 
A single individual eigenvalue can still be learned to relative precision $\epsilon$ using only $O\bigl(\epsilon^{-2}\bigr)$ measurements, however. 
It is the need to sweep through $2^n+1$ (or more) sets that leads to the factor of $O\bigl(n 2^n\bigr)$ in the sample complexity. 

In Ref.~\cite{Flammia2019}, the authors also derived what is essentially a variant of the Kushilevitz-Mansour algorithm~\cite{Kushilevitz1993} for learning decision trees via the Fourier spectrum and applied it to the case of Pauli channels. 
The idea is to breadth-first search through the marginal Pauli error rates, keeping those with large probability mass and pruning the search tree when the mass is below a threshold. 
This algorithm is theoretically efficient in $s$ and $n$, however our numerical experiments using this algorithm suggest that the number of eigenvalues required per recovered error rate will make it difficult to use in practice, at least in its current instantiation and in the relevant regime for quantum computing applications.

\section{Algorithm Overview}\label{sec:sparse}

The problem of reconstructing a sparse set of Pauli error rates by measuring few eigenvalues is closely related to a classical problem of computing a sparse Walsh-Hadamard transform.
This problem was studied by Scheibler \textit{et al}.~\cite{Scheibler2015} and later (in the regime of noisy signals) by Li \textit{et al}.~\cite{Li2015} by decoding a signal $x\in\mathbb{R}^N$ which contains $2^n$ points indexed by $j\in\mathbb{F}^n_2$. 
In our circumstances we are not analyzing the frequency domain of a signal, but rather the global probability distribution of the Pauli error rates in a quantum device and the eigenvalue distribution of the Paulis in a super-operator representation of a Pauli noise channel, so this formalism requires some adaptation. 

Given the WHT mapping in equations \eqref{eq:WHT} and \eqref{eq:reverse:WHT} the algorithms presented in \cite{Scheibler2015} and \cite{Li2015} are broadly applicable, but require some modifications.
We will note where adjustments have to be made.
One major difference is our inability to simultaneously measure noncommuting Pauli operators. 
Below we give a broad overview of the reconstruction algorithms as applicable to our needs.
A complete and rigorous analysis can be found in \cref{sec:proof}, but the main recovery guarantee is stated below in \Cref{Thm:main}.
We first deal with the noiseless case. 

The main idea behind the algorithm is to note that each eigenvalue is made up of a linear combination of all the error rates. 
By subsampling the eigenvalues, we are able to split up the error rates, figuratively creating `bins' of error rates, where each bin contains a linear combination of a smaller number of error rates. 
Provided that there are sufficient bins, then in the sparse regime most of these bins will only contain a few error rates with weight $\geq \epsilon_0$. 
Using \textit{aliasing}, we can identify these bins and can therefore evaluate these error rates. 
This information will allow us to reconstruct all the sparse error rates. 
With this in mind, the reconstruction algorithm can be broken down into three main steps:
\begin{enumerate}
    \item Determine the \textit{subsampling} bins and perform the experiments to measure the required eigenvalues.
    \item Calculate and measure the \textit{aliased} bins to enable identification of single Pauli-bins (\textit{single-tons}) and the Pauli error rates that occupy them.
    \item Run a decoder to `\textit{peel back}' single-tons, converting multi-Pauli bins to single-Pauli bins and repeat until all error rates are identified.
\end{enumerate}
We describe these three steps in an intuitive manner below and relegate the analysis and proofs to \cref{sec:algorithm}.

\paragraph*{Step 1---Subsampling.}
The intuition behind the first step is that it is possible to sample a specific pattern of eigenvalues that will allow the reconstruction of the global probability vector, but where various probabilities are binned (i.e.\ added together). 
For instance, given a global probability vector with $N = 4^n$ values it is possible to rewrite this as a ``reduced'' vector ($\boldsymbol{\tilde{p}}$) with $B=2^b$ values, each value being composed of the summation of $N/B$ of the original global probability values (possibly with signs). 
In the regime where our sparsity is $s < 2^n$ then we will show that with appropriate random sampling a large number of these reduced vector values (which we will call \textit{bins}), will be composed of none or one of our sparse Pauli errors, i.e.\ those with a weight $\ge \epsilon_0$ for a parameter $\epsilon_0$ to be chosen later.
In what follows we will always choose $B=2^n$, but we will occasionally use the notation $B=2^b$ (so that $b=n$) to illustrate where a given numerical factor originates from. 

Whereas Ref.~\cite{Scheibler2015} imagined using specific bit patterns of binary strings to index the requisite eigenvalues to sample, we wish to exploit the ability of a quantum device with independent measurement on each qubit to sample from a bit string of $2^n$ values. 
As previously discussed, the protocol in \cite{Flammia2019} shows how to measure, to multiplicative precision, the Pauli eigenvalues of $2^n$ commuting Paulis using one randomized-benchmarking style experiment with $n$-bit spin measurements at the output. 
The constraint that the Paulis measured be mutually commuting is exactly the constraint we require for the subsampling to allow us to create the required reduced probability vector $\boldsymbol{\tilde{p}}$.

Suppose we have a specific stabilizer group $\boldsymbol{\mathsf{S}}$. 
We will postpone how to choose this group until later. 
We can represent the entire stabilizer group by an $n \times 2n$ binary matrix $S$ whose $j^{\text{th}}$ row is the stabilizer generator $s_j$. 

Now let $v\in \mathbb{F}^n_2$ label the elements of the chosen stabilizer group, for example via the mapping $v.S$, where $v$ is thought of as a row vector. 
Our reduced probability vector $\boldsymbol{\tilde{p}}$ then consists of $B$ bins each containing a sum of $N/B$ distinct Pauli errors. 
It is labeled by a string $j \in \mathbb{F}_2^b$ and is given by:
\begin{equation}
    \tilde{p}_j = \frac{1}{B}\sum\limits_{v\in \mathbb{F}^n_2}\lambda_{v.S}(-1)^{j\cdot v}\,.\label{eq:quasiProbl}
\end{equation}
The effect of this is that the sampled Pauli eigenvalues from the stabilizer group, when transformed by the Walsh-Hadamard transform, give us $B$ bins each containing a sum of $2^n = N/B$ error rates, many of which will be zero in general. 

The binning is chosen in such a way that with high probability there will be a large number of bins that only contain a \textit{single} Pauli error rate with a weight $\ge \epsilon_0$ (the other Pauli errors allocated to that bin being, effectively, zero).
This will depend on the size of the bins and the sparsity of the Pauli error rates, and is discussed further in \cref{sec:algorithm}. 
An simple example of the subsampling and binning idea is shown in \Cref{fig:bipartite}. 

So how do we construct our stabilizer group? 
The most obvious way is to sample a random $n$-qubit Clifford (see \cite{Koenig2014,bravyi2020hadamardfree} for how to do this). 
However as $n$ grows past a few qubits, then on current devices the number of single and multi-qubit gates required to construct a generic element of the Clifford group requires circuits of depth $O(n^2/\log n)$, and if these circuits are noisy then this will wash out the signal required to estimate the eigenvalues. 
A better way for current devices is to use a random subset of $n$-qubit stabilizers that can be formed from a single round of non-overlapping 2-qubit Clifford gates. 
This has the added advantage of making it trivial to work out how to perform step 2.

\begin{figure*}[thb]
    \includegraphics[width=1.8\columnwidth]{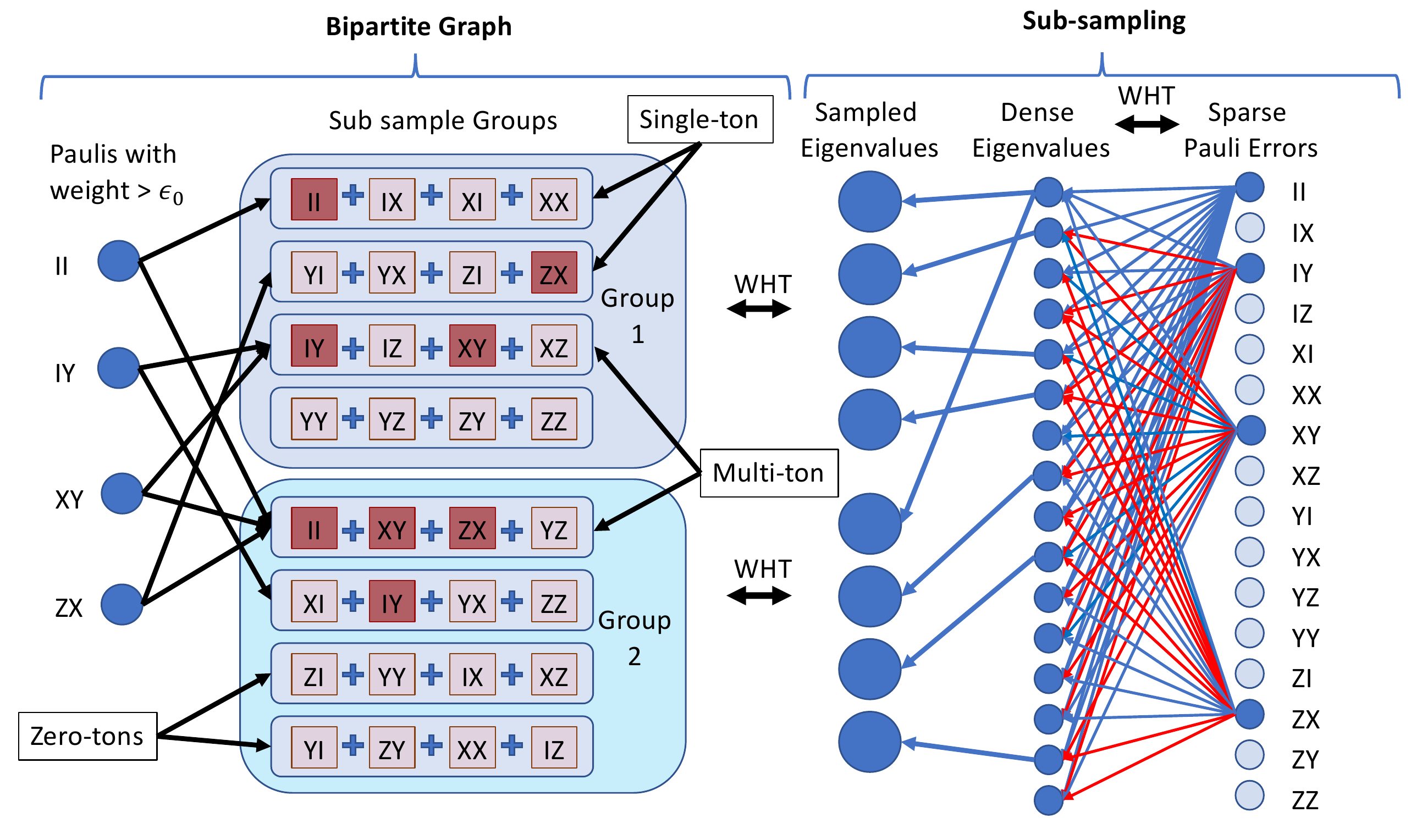}
    \caption{Illustrative diagram of the bipartite graph that is used to extract information from the subsampling bins. 
    Here we show a simple example for two qubits, where only three non-trivial Paulis have errors (shown in dark red). 
    The subsampling algorithms split the Paulis up as shown ($II$, $IX$, $XI$ and $XX$ stabilizers for group 1 and $II$, $ZY$, $XZ$ and $YX$ for the group 2), separating them into bins consisting of singletons, multi-tons and zero-tons, as described in the text. 
    In this example it can be seen that the $IY$ Pauli exists as a single-ton in group~2, allowing its value to be recovered. 
    It can then be `peeled' from the third bin in group~1 converting that bin from a multi-ton to a single-ton containing just Pauli~$XY$. 
    This then allows Pauli~$XY$ to be recovered as well, which would not otherwise be possible, as the signal would conflate with that of~$IY$ from group~1 alone.
    Iterative peeling in this fashion will eventually recover all of the nonzero Pauli error rates. 
    On the right side of the figure we illustrate how the sub-sample groups are formed. 
    The sparse Pauli errors are transformed into the dense eigenvalues by the WHT transform. 
    Here we have only added the lines relating to the nonzero Paulis, with blue a positive and red a negative contribution to the relevant eigenvalue. 
    The sampled eigenvalues are transformed to define the subsamping groups. 
    }\label{fig:bipartite}
\end{figure*}

\paragraph*{Step 2---Aliasing.}
The question then becomes: how do we detect which bins contain a single Pauli error rate? 
To do this the reconstruction algorithm uses the shift/modulation property of the WHT.
Specifically, if we let $\{p_k\}$ be the WHT of $\{\lambda_m\}$ we have:
\begin{equation}
    \lambda_{m+n}\, \overset{\text{WHT}}\longleftrightarrow\, (-1)^{\langle n,k \rangle}p_k\,.
\end{equation}
By taking each element of the stabilizer group and offsetting the sample with a shifting bit pattern (e.g.\ for four qubits the sample would be offset by the five following  bit patterns $[0,0,0,0]$, $[1,0,0,0]$, $[0,1,0,0]$, $[0,0,1,0]$, $[0,0,0,1]$) then the Pauli error rates consigned to that bin are no longer merely summed but rather are added or subtracted depending on whether the inner product of their `bit-strings' and the relevant pattern is zero or one. 
This result will be illustrated in more detail in \Cref{alg:subsampling} and \Cref{lm:prop_hashing_obs}, where we also discuss how to use bit-flip error detection codes to make the decoding more robust to noise. 

This leads to a number of remarkable effects. 
If the bin is empty (i.e.\ contains no Pauli error rates with non-zero errors) each of the \emph{offset bins} (i.e.\ for a particular $j$ each $\tilde{p}_{j,d}, d\in \{2^0\ldots2^{2n}\}$) will also be zero. 
If the bin contains only one non-zero Pauli error rate then the magnitude of the sum of each of the offset bins will be constant, and the sign of the sums will identify exactly which Pauli has the non-zero error rate. 
(For example using the four qubit offsets shown above, if the absolute values of the bins were all $0.001$ and the signs of the 4 offset bins were $(+,-,-,+)$, this could only be caused by a single Pauli error rate of $0.001$, with a bit string of $0110$). 
In every other case, the bin contains multiple Pauli error rates (a \textit{multi-ton bin}), which leads us to the PEELING decoder (see step 3).

So how can we construct the experiments that will allow us to extract the `shifted' eigenvalues? 
For instance, one might note that for any particular stabilizer group $\mathcal{S}$, the offset bit pattern applied to each of the elements of the group are unlikely to form a stabilizer group.

It transpires that where we use a stabilizer group created by local two-qubit Cliffords (on each qubit pair), we can do this simply by iterating each distinct qubit pair through four further (different) two-qubit stabilizer patterns. 
There are five two-qubit stabilizer groups, the union of whose bases form a complete set of mutually unbiased bases; let us label them $S^{\otimes 2}_{1\dots5}$. 
We set out a specific choice of these groups in detail, together with the two-qubit circuit needed to create them, in \Cref{fig:circuits}. 
The initial stabilizer is chosen by selecting randomly from $S^{\otimes 2}_{1,2}$ for each qubit pair.
This becomes the stabilizer for the purpose of Step 1. 
This circuit is used to conduct the first experiment and extract $2^n$ Pauli eigenvalues. 
The offset pattern required for this in Step 2 is constructed by iterating over each qubit pair, and replacing the circuit chosen in Step 1 with one of the other 5 (for a total of 4 further experiments per qubit pair). 
The total number of experiments required is therefore $2n+1$. 
By analyzing each of the experiments formed we will be able to pull out of the all of the eigenvalues determined by such experiments. 
\Cref{fig:circuits} shows the circuits used for each experiment and illustrates the method described above.

\paragraph*{Step 3---Peeling.}
If we use a variety of subsampling matrices (that is we repeat Steps 1 and 2 for more than one random initial choice of Cliffords) we are now in the position where we have identified a number of Pauli error rates (from bins that contain only one Pauli error rate) and we will also have a number of bins that contain more than one Pauli error rate (multi-ton bins). 
In general, for any two stabilizer groups, different Pauli error rates will get hashed into different bins. 
Where we have identified a single Pauli error rate under, say, stabilizer group 1, that same error rate may be in a different bin under stabilizer group 2, a bin it shares with one or more different high weight Paulis (i.e.\ it may be in a multi-ton bin under stabilizer group 2). 
However, because we know the value of this Pauli error rate (since it was a \textit{singleton} under stabilizer group 1), we can remove it from the bin created by stabilizer group 2 by simple subtraction. 
After this removal, some bins that were previously multi-ton bins will now become singletons, or at the very least they will be closer to being singleton in that we are left with a bin that now has one fewer Pauli error rate in it.
This removal of the value of a previously identified singleton from a different stabilizer group's bin is known as `peeling back' the known values, giving the PEELING decoder its name. 
The goal is that when we peel back our identified error rates, we create more and more bins that now contain only one Pauli error rate. 
This can be applied in an iterative fashion. 
We can then iterate this until we have either identified all the Pauli error rates (all the bins are empty) or until we have no further single Pauli error rates to peel back.
All of these steps can be viewed in \Cref{alg:peeling}. 
In the latter case the reconstruction algorithm has failed, although we will at least know the magnitude of the error rates we have failed to identify, and can perform additional experiments to try to learn them.

\subsection{Dealing with noise}

Using the ideas in Ref.~\cite{Li2015}, we can modify the reconstruction algorithm to handle noise of the form 
\begin{equation}
    \boldsymbol{\lambda} \to \boldsymbol{\lambda} + \boldsymbol{w}\,,
\end{equation}
where $\boldsymbol{w}$ is a Gaussian distributed noise vector, $\boldsymbol{w}\sim\mathcal{N}(0,\xi^2 \mathbbm{1})$. 
It is only for simplicity in the proof that we consider the isotropic case, and small dependencies and correlations do not substantially affect the observed numerical performance.

In our case, the noise arises as the estimation error in our eigenvalues caused by finite sampling. 
These finite sampling errors occur because of the limited number of random sequences and measurement shots per sequence occurring when the eigenvalue estimation experiments are carried out.
Errors of this nature have been analyzed and in Ref.~\cite{Flammia2019}.
To reduce noise, the number of sequences and shots per experiment needs to be increased, and this sample complexity was also bounded in Ref.~\cite{Flammia2019}.
In the relevant regime of high precision, the estimation error on the eigenvalues will be approximately normally distributed, and empirical estimates of the variance and covariance can be determined by bootstrapping from the observed measurement outcomes~\cite{Harper2019}.

The PEELING decoder only requires two adjustments to account for such noise: the zero Pauli verification and the single Pauli search protocols. 

For the former, in the noiseless model we identify a bin as being empty if the value of the bin (and each of the offset bins) is zero. 
Where we have noise, we simply relax the requirement that the bins are exactly equal to zero before identifying them as empty.
We can bound an acceptable small value as indicating an empty bin, given the number of `noisy' zeros in the bin and our estimate of the noise variance. 
This will lead to a noise floor of Pauli error weights we can recover.
That is, we are unlikely to recover those Pauli errors with a value so small they are swamped by the noise in the bins.
This is an inevitable consequence of the noise.

The latter case of single Pauli identification has two aspects that need to be considered.
The first is `does the bin contain only a single Pauli?', and the second is `if so: which Pauli?'. 
For a noisy version the first question is dealt with the same way as the noisy zero, i.e.\ we only require the magnitudes of the offset bins to match to within some estimated noise window. 
While this runs the risk of not noticing some small Pauli error rates that are also in the bin, it appears to work well in practice.
The second is more akin to a noisy bit flip channel, in that the noise may cause us to incorrectly identify a `1' as a zero or vice-versa. 
(This is more likely when the noise is commensurate with or greater than the Pauli error weight.) 
One simple method of dealing with this is to repeat sample with different offsets, and then take a majority vote, however our numerical simulations do not suggest that this is necessary. 
Finally we can use a number of random offsets and some additional fixed offsets chosen in such a way they form a classical error correction code to further protect the algorithm from noise. 
When an appropriate classical code is chosen this does not alter the sample complexity scaling, though it does increase slightly the number of experiments.
It also comes with a robust recovery guarantee as described in the next section and \cref{sec:algorithm}.

\section{Recovery guarantee from noisy eigenvalues}\label{sec:guarantee}

Using the algorithm illustrated above and leveraging some proofs contained in \cite{Li2015}, we can construct the following recovery guarantee that relates our ability to recover Pauli error rates with bounded error to the noise in the estimated Pauli eigenvalues. 
The intuition behind the guarantee is that by increasing the number of offset observations we can reduce the chance of incorrectly detecting whether the bin occupancy is zero, or one, or more than one. 
If the bin detection succeeds, then the peeling step will succeed with high probability for appropriate choices of the subsampling and aliasing designs. 

Our recovery guarantee does however rely on several assumptions, which we now state explicitly.
\begin{assumption}\label{random_support_assumption}
\rm Let $\boldsymbol{p}\in\mathbb{R}^{N}$ be the target Pauli error rates with support $\mathcal{K}=\mathrm{supp}(\boldsymbol{p})$ and sparsity $s = |\mathcal{K}|$. 
\begin{itemize}
	\item[$\bf A1$] \textit{(Random sparse support.)} The support set $\mathcal{K}$ is chosen uniformly at random from all subsets of $[N]$ of size exactly $s$, where $s = 4^{\delta n}$ is sub-linear in the dimension $N = 4^n$ for some $0<\delta<1/2$. 
	\item[$\bf A2$] \textit{(Independent Gaussian noise.)} Each queried Pauli eigenvalue $\lambda_j$ has noise given by independent Gaussian noise centered around the eigenvalue with variance $\xi^2$. 
	\item[$\bf A3$] \textit{(Good signal-to-noise.)} Each error rate $p_m$ for $m\in\mathcal{K}$ is lower bounded by $p_m \ge \epsilon_0$ for some $\epsilon_0 > 0$, the eigenvalue noise variance is upper bounded as $\xi^2\leq\min(\frac{B}{s^2},1)$, and the two are related via $\epsilon_0 \geq 2\xi/\sqrt{B}$. 
	Here $B$ is the number of bins in a single subsampling group. 
\end{itemize}
\end{assumption}

Our main theorem is then the following.
\begin{theorem}\label{Thm:main}
Suppose the Assumptions \ref{random_support_assumption} hold for an unknown Pauli channel with eigenvalues $\boldsymbol{\lambda}$ and error rates $\boldsymbol{p}$. 
Then with failure probability $\mathbb{P}_F\leq \mathrm{e}^{-O(n)}$, Algorithms~\ref{alg:subsampling}, \ref{alg:peeling}, and \ref{alg:bin_detect} estimate the $s$-sparse Pauli error rates $\widehat{\boldsymbol{p}}$ such that $\|\widehat{\boldsymbol{p}}-\boldsymbol{p}\|_{\infty}\leq 2\xi/\sqrt{B}$ using $O\bigl(s n\bigr)$ eigenvalue queries and $O(sn^2)$-time classical computation.
\end{theorem}
\begin{proof}
The proof is given in \cref{sec:proof}.
\end{proof}

Note that our main theorem references a noisy eigenvalue oracle rather than a direct sample complexity for estimating the eigenvalues. 
From~\cite{Flammia2019}, $O(sn)$ queries to the eigenvalue oracle can be approximated to within variance $\xi^2$ using only $O\bigl(\frac{n^2}{\xi^2}\bigr)$ samples. 
While a variant of the protocol in~\cite{Flammia2019} can make the noise independent, it will not be exactly isotropic Gaussian noise, so we can only heuristically claim this as the sample complexity. 
This is why we state the formal main result in terms of query complexity. 

It is worth remarking on the strength of the assumptions that go into the statement of the theorem. 
Assumption A1 is mathematically convenient, but is certainly too strong physically since most errors in near-term quantum devices are likely to have low weight. 
This could in principle be compensated by incorporating a randomizing permutation into the experimental design. 
However, our experiments (see the next section) do not seem to require such a compensation for convergence.
Assumption A2 is again mathematically convenient, and it will only ever be approximately true in practice. 
We believe that other error models with weak correlations and bounded variance will have similar guarantees, but an analysis of this would introduce significant complications without elucidating anything about the algorithm. 
Weakening A2 in this way would be interesting future work, as it would let us make direct formal statements about the sample complexity.
As for our final assumption, a signal-to-noise assumption along the lines of Assumption A3 seems to be a mathematical necessity for convergence. 
However, it may be possible that a guarantee could still be proven with a smaller signal-to-noise ratio or with weaker restrictions on $\epsilon_0$ and $\xi$. 
For example, a simple corollary of our result is a guarantee in the total variation distance (1-norm) such that $\frac12\|\widehat{\boldsymbol{p}}-\boldsymbol{p}\|_{1}\leq s\xi/\sqrt{B}$, which is nontrivial exactly when $\xi < \sqrt{B}/s$ (cf.\ A3). 
It might be easier (and more natural) to directly prove this implication of our result, or it may be possible to prove this using weaker assumptions.

\section{Experimental validation}\label{sec:validation}

To validate our algorithm we use data extracted from a 14-qubit superconducting device build by IBM. 
In Ref.~\cite{Harper2019} the complete distribution of locally averaged Pauli error rates in the device was estimated. 
In this work, we recycle the data from that experiment to validate our new algorithms. 

The data set from Ref.~\cite{Harper2019} consists of $2^{14}$ \textit{locally averaged} eigenvalue estimates, meaning that each eigenvalue is labeled by a $14$-bit string that labels the presence or absence of a nontrivial Pauli on each corresponding qubit. 
This is in contrast to the full eigenvalues, each of which would require a $28$-bit label and could additionally resolve the entire set of $2^{28}$ Pauli eigenvalues, without local averaging. 
Although we could run our algorithm on the $2^{14}$ locally averaged eigenvalues, to make it more challenging we have looked at random self-consistent extrapolations of the data onto the full set of $2^{28} \approx 2.7\times10^{8}$ eigenvalues.

The random interpolation proceeds as follows. 
From the estimated eigenvalues of Ref.~\cite{Harper2019}, we reconstruct the locally averaged error rates. 
(In fact, this step was already done in~\cite{Harper2019}.) 
For each locally averaged error rate, we pick a uniformly random point in the probability simplex of the Paulis supporting the local average. 
This defines a new probability distribution on the full set of $2^{28}$ Paulis.
Every such extrapolation has the property that locally averaging it will return the original experimentally observed data. 
We construct a ``true'' set of known Pauli channel eigenvalues by transforming (using the Walsh-Hadamard transform) on these extrapolated error rates. 
This gives us a family of experimentally derived eigenvalue oracles that we can use to validate our numerical reconstructions.

The data from Ref.~\cite{Harper2019} have a `no-error' probability of about 0.86, and upon extrapolation they have approximately 200 Paulis with an error rate above $10^{-5}$, about 600 above $10^{-6}$, and about 2,000 above $10^{-8}$. 
Although the original estimation cannot resolve error rates as small as $10^{-8}$ with meaningful error bars, our eigenvalue oracle still has access to these numbers as part of the simulation. 
For this discussion, we will focus on reconstructing errors in the regime above $10^{-5}$, as these are the most relevant. 
This corresponds to a sparsity $s = 4^{\delta n}$ with roughly $\delta\approx\frac14$.

As can be seen from \Cref{fig:initial_results}, the sparse recovery protocol performs well in this regime ($\delta \lesssim 0.25$), requiring only a fraction of the eigenvalues that would be required for a full recovery of all Pauli error rates. 
The limiting factor in this regime is the noise in the oracle, which equates directly to the number of measurements and sequences sampled as part of the original experiment (see \cite{Flammia2019} for relevant reconstruction guarantees). 
It appears that the effect of the protocol is to allow recovery of the Pauli error rates to (approximately) an order of magnitude or more \textit{less} than the noise in the oracle. 

Importantly, if a device has unexpected many-body correlations (for example through unexpected qubit interactions or crosstalk), then we should also be able to find these errors whenever their probability is above our noise floor. 
We have validated this feature of the algorithms as well by injecting known high-weight Pauli errors into the oracle. 
Our algorithm reveals and evaluates such Pauli error rates to a high degree of relative accuracy, as shown in \Cref{fig:initial_results}(e).

Section \ref{sec:additional} discusses the regime where ($\delta \geq 0.25$). 
In that case continued recovery of Paulis with low error rates requires some changes to the local stabilizer groups used (or a switch to global random stabilizer groups).

\subsection{Experiments in the regime \texorpdfstring{$\frac{1}{4} < \delta < \frac{1}{2}$}{0.25 < delta < 0.5}}\label{sec:additional}

While the experimental protocol presented above is likely to be all that is required in most practical regimes, if the number of Paulis to be recovered is large then a slight modification might be needed. 
Unlike the situation where one is using completely random stabilizer groups, the local stabilizer protocol can fail when trying to reconstruct many low-error Paulis that differ only in one or two Paulis, in such a way that they cannot be separated by the local stabilizers.
This might occur, for instance, in the regime where $\delta > 0.25$. 
In such circumstances, one can cycle each distinct set of two qubit pairs through the 5 stabilizer groups identified in \Cref{fig:circuits}(c), and then generate the offset bins for each of them. 
The number of experiments that need to be performed are the original experiment (1), then a further 4 for each qubit pair (4), times the number of qubit pairs ($n/2$), times the number of experiments needed to generate the offsets on the remaining $n-2$ qubits ($4(n-2)$), for a total of $1+8n(n-2)=O(n^2)$ total experiments. 
The eigenvalues gathered this way allow the creation of $n/2$ properly offset subsampling matrices of $2^{n+2}$ bins each containing $2^{n-2}$ Paulis. 
Empirically, this appears to be sufficient to exactly recreate the global probabilities up to $\delta=0.5$. 
\Cref{fig:initial_results}(b) illustrates the extra recovery power available in the highest precision regime.

\section{Heuristic noise reconstruction}\label{sec:heuristic}

Here we describe in more detail the intuition behind the algorithm, the experiments prescribed and a simplified, practical extraction algorithm. 
Our GitHub repository~\cite{GITHUB} contains code and examples showing how the algorithms can be used to recreate the figures in this paper.

\subsection{Determining a suitable number of subsampling groups}\label{sec:bounds}

Our proofs relating to the recovery of $s$-sparse Pauli errors require an assumption that each element in the support set $\mathcal{K}$ is chosen independently and uniformly at random from $[N]$. 
At first glance it may appear that this is not likely to be the case in a quantum device as the Pauli errors are likely to cluster around low-weight Pauli errors rather than be uniformly distributed over the $4^n$ different possible Pauli errors. 
However where we choose random $n$-qubit stabilizers (global stabilizer groups) as the basis for sampling the Pauli eigenvalues, this effectively randomizes the bin into which we consign any specific Pauli error rate, which (empirically) allows us to satisfy the uniformly random distribution requirement.

Given this we can continue to use the ``balls-and-bins'' model utilized in \cite[Appendix B]{Li2015}.
We can use this insight, together with our ability to simultaneously sample $2^n$ commuting eigenvalues, to determine practical values for the number of subsampling groups $C$ given our bin size of $B=2^n$.

Since the sparsity $s\ll N$, we have that the expected number of Paulis (balls) in one bin will be $\frac{s}{N}\times \frac{N}{B} = \frac{s}{B}$, which in the sparsity regime of interest will be $< 2$. 
Assuming we have an $s$-sparse distribution, with $B=2^n$ being the number of bins sampled, we define the sparsity coefficient $\eta = \frac{B}{s}$, which is at least 1. 
This means that we only require $C$, the number subsampling groups, to be $2$ in order to recover all of the edges in $O(s)$ iterations with probability at least $1-O(1/s)$ (see \cite[Appendix B]{Li2015}). 
It can then be seen that as each experimental run recovers $2^n$ eigenvalues, we need to perform at least one experimental run for each bin plus one for each of the offsets of the bin times the number of subsampling groups.
This means that the minimum number of experimental runs is $2(2n+1)$. 
As we discuss later, by increasing the number of offsets we can increase the recovery guarantees, but at the cost of sampling more eigenvalues (although still only scaling proportional to $n$). 

In the case where our device is too large for $\eta \geq 1$, for instance where we have had to marginalize over the measurements as $\log_2(B)\geq 30$, then we can increase our effective $C$ by marginalizing over randomly chosen qubits and creating our subsampling matrices from such randomly chosen sub-samples of the measurement outcomes. 
This will allow us to retain the recovery guarantees without increasing the number of experiments on the device, although this incurs an increased computational cost in setting up and performing the peeling decoder.

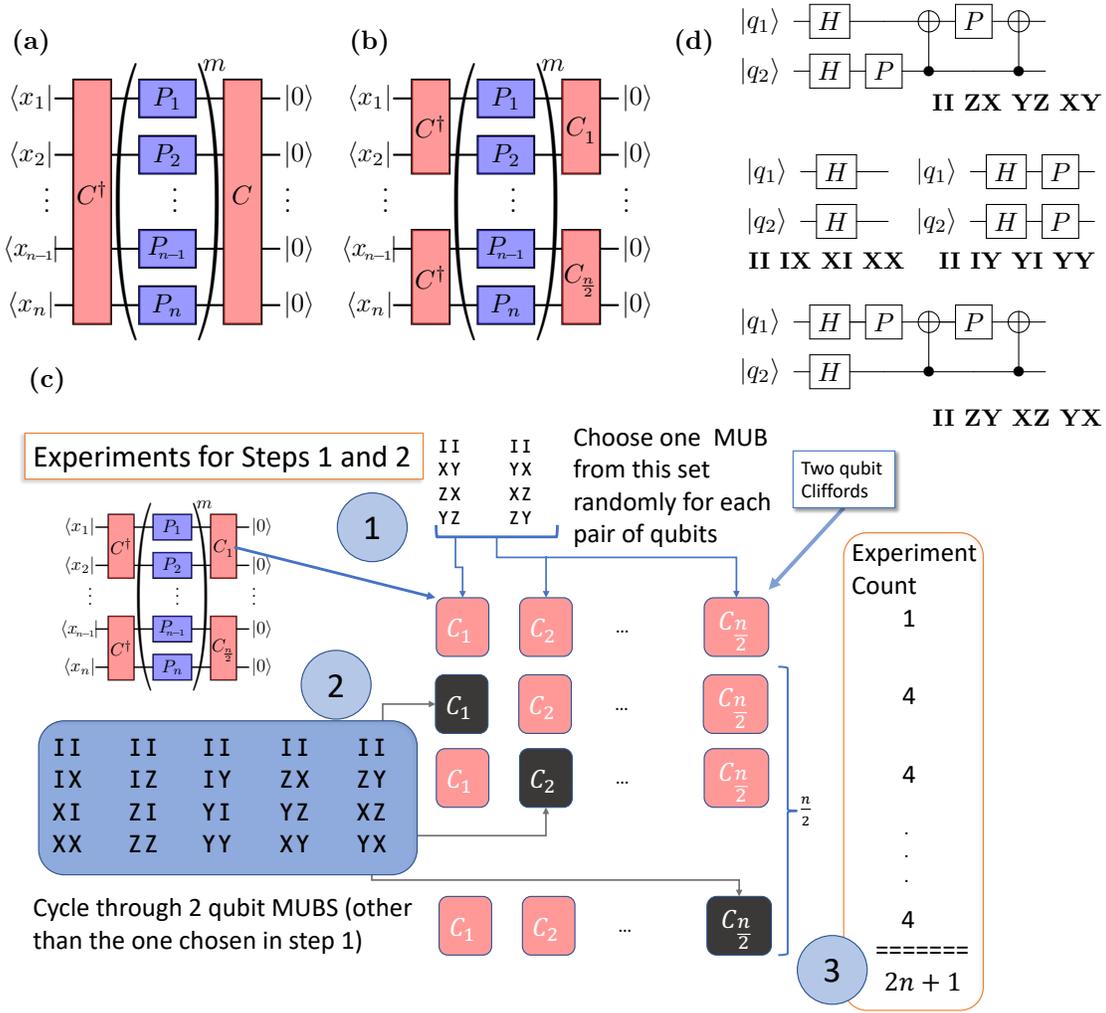
\begin{figure*}
    \centering
 \begin{tikzpicture}
	[w/.style={fill=white,thick},
	b/.style={fill=black,thick},
	p/.style={fill=white!20!blue,fill opacity = .5},
	c/.style={fill=white!20!red,fill opacity = .5},
	h/.style={fill=white!20!green,fill opacity = .5},
    i/.style={fill=white!10!red,fill opacity = .6}]

	\begin{scope}[shift={(0,0)}]
	\foreach \y in {1,...,2} {
		\draw[thick] (0,3-0.75*\y) -- (3,3-0.75*\y);
		\filldraw[w] (9/8,2.75-0.75*\y) rectangle (15/8,3.25-0.75*\y);
		\fill[p] (9/8,2.75-0.75*\y) rectangle (15/8,3.25-0.75*\y);
		\node at (3/2,3-0.75*\y) {$P_{\y}$};
		\node at (26/8,3-0.75*\y) {$|0\rangle$};
		\node at (-5/16,3-0.75*\y) {$\langle x_{\y}|$};
	}
	\foreach[evaluate=\y as \o using int(4-\y)] \y in {3,...,4} {
		\draw[thick] (0,2.5-0.75*\y) -- (3,2.5-0.75*\y);
		\filldraw[w] (9/8,2.75-0.75*\y) rectangle (15/8,2.25-0.75*\y);
		\fill[p] (9/8,2.75-0.75*\y) rectangle (15/8,2.25-0.75*\y);
		\node at (3/2,2.5-0.75*\y) {\ifthenelse{\o=0}{$P_n$}{$P_{\!\scriptscriptstyle{\!n\!-\!\o}}$}};
		\node at (26/8,2.5-0.75*\y) {$|0\rangle$};
		\node at (-5/16,2.5-0.75*\y) {\ifthenelse{\o=0}{$\langle x_n|$}{$\langle x_{\scriptscriptstyle {\!n\!-\!\o\!}}|$}};
	}
	\node at(-5/16,3) {\textbf{(a)}};
	
	\filldraw[w] (1/4,-3/4) rectangle (3/4,5/2);
	\fill[c] (1/4,-3/4) rectangle (3/4,5/2);
	\filldraw[w] (9/4,-3/4) rectangle (11/4,5/2);
	\fill[c] (9/4,-3/4) rectangle (11/4,5/2);
	\node[yscale=3.5,xscale=1.5] at (0.9,7/8) {$\Biggl($};
	\node[yscale=3.5,xscale=1.5] at (2.1,7/8) {$\Biggr)$};

	\node at (1/2,8/8) {$C^\dagger$};
	\node at (5/2,8/8) {$C^{\vphantom{\dagger}}$};
	\node at (2+2/16,43/16) {$m$};

    \node at (-0.1,1){\vdots};
	\node at (1.6,1){\vdots};
	\node at (3.1,1){\vdots};

\end{scope}

\begin{scope}[shift={(4.5,0)}]
	\foreach \y in {1,...,2} {
		\draw[thick] (0,3-0.75*\y) -- (3,3-0.75*\y);
		\filldraw[w] (9/8,2.75-0.75*\y) rectangle (15/8,3.25-0.75*\y);
		\fill[p] (9/8,2.75-0.75*\y) rectangle (15/8,3.25-0.75*\y);
		\node at (3/2,3-0.75*\y) {$P_{\y}$};
		\node at (26/8,3-0.75*\y) {$|0\rangle$};
		\node at (-5/16,3-0.75*\y) {$\langle x_{\y}|$};
	}
	
	\foreach[evaluate=\y as \o using int(4-\y)] \y in {3,...,4} {
		\draw[thick] (0,2.5-0.75*\y) -- (3,2.5-0.75*\y);
		\filldraw[w] (9/8,2.75-0.75*\y) rectangle (15/8,2.25-0.75*\y);
		\fill[p] (9/8,2.75-0.75*\y) rectangle (15/8,2.25-0.75*\y);
		\node at (3/2,2.5-0.75*\y) {\ifthenelse{\o=0}{$P_n$}{$P_{\!\scriptscriptstyle{\!n\!-\!\o}}$}};
		\node at (26/8,2.5-0.75*\y) {$|0\rangle$};
		\node at (-5/16,2.5-0.75*\y) {\ifthenelse{\o=0}{$\langle x_n|$}{$\langle x_{\scriptscriptstyle {\!n\!-\!\o\!}}|$}};
	}
	
	\foreach[evaluate=\y as \o using int(1+\y)] \y in {0,...,0} {

		\filldraw[w] (9/4,1.25-2*\y) rectangle (11/4,2.5-2*\y);
		\fill[c] (9/4,1.25-2*\y) rectangle (11/4,2.5-2*\y);
		\node at (5/2,1.875-2*\y) {$C^{\vphantom{\dagger}}_{\o}$};

		\filldraw[w] (1/4,1.25-2*\y) rectangle (3/4,2.5-2*\y);
		\fill[c] (1/4,1.25-2*\y) rectangle (3/4,2.5-2*\y);
	    \node at (1/2,1.875-2*\y) {$C^{\dagger}$};
		
	}
	\foreach[evaluate=\y as \o using int(1+\y)] \y in {1,...,1} {

		\filldraw[w] (9/4,1.25-2*\y) rectangle (11/4,2.5-2*\y);
		\fill[c] (9/4,1.25-2*\y) rectangle (11/4,2.5-2*\y);
		\node at (5/2,1.875-2*\y) {$C^{\vphantom{\dagger}}_{\!\frac{n}{2}}$};

		\filldraw[w] (1/4,1.25-2*\y) rectangle (3/4,2.5-2*\y);
		\fill[c] (1/4,1.25-2*\y) rectangle (3/4,2.5-2*\y);
	    \node at (1/2,1.875-2*\y) {$C^{\dagger}$};
		
	}
	
	\node at(-5/16,3) {\textbf{(b)}};
    \node at (-0.1,1){\vdots};
	\node at (1.6,1){\vdots};
	\node at (3.1,1){\vdots};

	\node[yscale=3.5,xscale=1.5] at (0.9,7/8) {$\Biggl($};
	\node[yscale=3.5,xscale=1.5] at (2.1,7/8) {$\Biggr)$};

	\node at (2+2/16,43/16) {$m$};
\end{scope}

	\node[inner sep=0pt] at (6,-6) {\includegraphics[width=1.5\columnwidth]{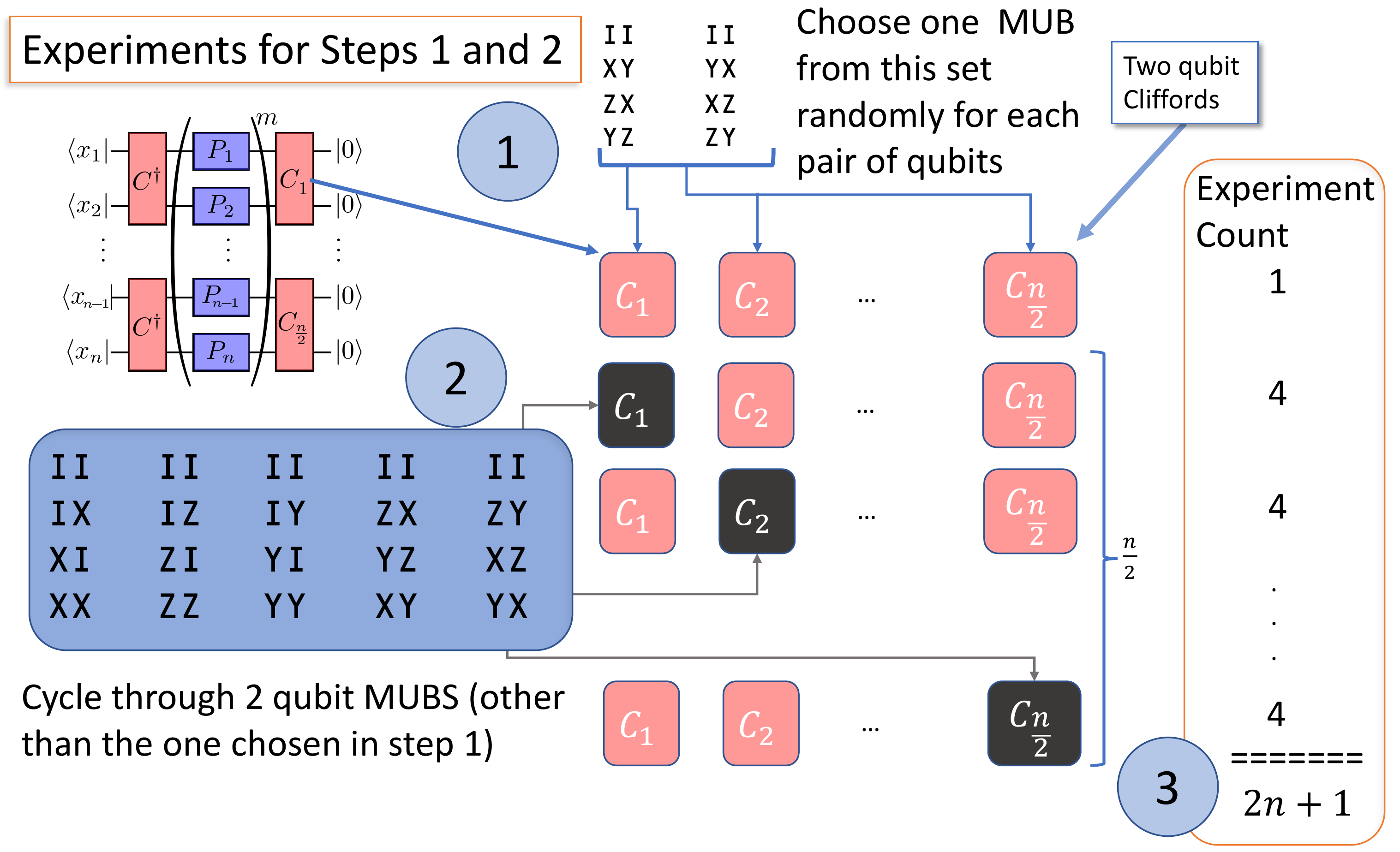}};
	\node at (-0.1,-1.5) {\textbf{(c)}};
	\node at (8.5,3) {\textbf{(d)}};
	
	\node[inner sep=0pt] at (11.5,3) {$\begin{array}{c}
			\Qcircuit @C=.6em @R=.6em {
  			\lstick{|q_1\rangle}	& \gate{H} &\qw& \targ & \gate{P}	& \targ & \qw \\
 			 \lstick{|q_2\rangle}	& \gate{H} &\gate{P} & \ctrl{-1} & \qw & \ctrl{-1}& \qw \\
			}
			\end{array} $};
	\node[inner sep=0pt] at (12.8,2.2) {\small{\textbf{II ZX YZ XY}}};
	\node[inner sep=0pt] at (10.5,1) {$\begin{array}{c}
		\Qcircuit @C=.6em @R=.6em {
		  \lstick{|q_1\rangle}	&\gate{H}	& \qw & \qw \\
		  \lstick{|q_2\rangle}	& \gate{H}	& \qw & \qw \\
		}
		\end{array}$};
	\node[inner sep=0pt] at (10.2,0.1) {\small{\textbf{ II IX XI XX  }}};
	\node[inner sep=0pt] at (11.5,-1) {$\begin{array}{c}
		\Qcircuit @C=.6em @R=.6em {
		  \lstick{|q_1\rangle}	& \gate{H} &\gate{P}\qw& \targ &\gate{P}\qw	& \targ & \qw \\
		  \lstick{|q_2\rangle}	& \gate{H} &\qw & \ctrl{-1} & \qw & \ctrl{-1}& \qw \\
		}
		\end{array}$};
		\node[inner sep=0pt] at (12.8,0.1) {\small{\textbf{II IY YI YY }}};
	
	\node[inner sep=0pt] at (13,1) {$\begin{array}{c}
		\Qcircuit @C=.6em @R=.6em {
		  \lstick{|q_1\rangle}	& \gate{H}  & \gate{P}	& \qw \\
		  \lstick{|q_2\rangle}	&\gate{H}  & \gate{P}	& \qw \\}
		\end{array}$};
		\node[inner sep=0pt] at (12.8,-2) {\small{\textbf{II ZY XZ YX }}};
	
\end{tikzpicture}	
    \caption{\textbf{(a)} shows the type of circuit described in \cite{Flammia2019} that allows the recovery of $2^n$ eigenvalues of the averaged noise channel of the device.
    The random Pauli gates (blue) are used to twirl the channel and by averaging over a number of random choices of Pauli, the noise channel in the device is transformed into a Pauli channel. 
    In a similar way to randomized benchmarking, by repeating the twirl for a certain number of Pauli gates ($m$) and then returning the system into the computation basis by choosing the Pauli inverting the twirl, a decay curve will be induced and this can then be fit to determine the eigenvalues independently of state preparation and measurement errors. 
    The single $n$-qubit Clifford (and its inverse at the end) determines which of the $4^n$ Paulis are sampled and an appropriate Clifford can be used to select any $n$-qubit stabilizer set. 
    \textbf{(b)} shows a further modification of the circuit, where instead of using a generic $n$-qubit Clifford only two-qubit Cliffords are used. 
    (Where the device has an odd-number of qubits a single Clifford can be used on one of the qubits.) 
    As discussed in the text, for each chosen value of $m$ the circuit is repeated for multiple sequences with different randomly chosen Paulis, but for fixed Cliffords. 
    Collectively each of the runs for multiple choices of Paulis carried out over several different lengths of $m$ are defined as \emph{an experiment}. 
    \textbf{(c)} shows how once an experiment has been chosen in Step 1 of the procedure, further experiments are created in order to determine the offsets required to identify the Pauli (see text). 
    As shown in Step 2, each of the two qubit Cliffords needs to be cycled sequentially through the four other 2-qubit stabilizer groups (i.e.\ the four that are different from the initial choice). 
    This means that for each sequence in Step 1, a further $2n$ experiments need to be performed, leading to $2n+1$ experiments per chosen stabilizer group. 
    For the second group an offset of one qubit should be chosen, meaning the local stabilizer groups now span different qubit pairs. 
    Simple Pauli twirls can be carried out on any odd or isolated qubits.
    \textbf{(d)} Some example sub-circuits required to perform the transform into the local stabilizer group listed in (c)-Step 2. 
    The inverse gate will be of a similar form.}
    \label{fig:circuits}
\end{figure*}

\subsection{Heuristic algorithm using local circuits}

Our numerical simulations based on the data collected in the experiments from Ref.~\cite{Harper2019} indicate that randomly chosen global stabilizers are not in fact necessary to distribute the Pauli errors widely enough to allow recovery. 
It appears that \textit{local} stabilizer groups suffice. 
This allows us to dramatically reduce circuit complexity while keeping the number of experiments required to a minimum. 
\Cref{fig:circuits}(c) details the local two-qubit stabilizer groups that can be selected to perform an extraction experiment that is viable on most current devices. 
In \Cref{fig:circuits}(d) we show the local Clifford circuits that can be used (with their corresponding inverse)  at the beginning (end) of the measurement circuits to create these local stabilizer groups. 
In Ref.~\cite{Flammia2019} it was shown that using such circuits we can estimate $2^n$ Pauli eigenvalues with relative precision $\epsilon$, using $O(\epsilon^{-2}n)$ measurements.

Having chosen the series of stabilizers to measure, together with the circuits for offsets, we will have all the relevant eigenvalues required to use the noisy peeling decoder. 

In \Cref{alg:NoisyDecoder} we show how to operate the decoder on a practical level, assuming that there are a large number of Paulis sitting below the level of interest (i.e.\ with an error rate $\ll \epsilon_0$). 
To understand how it works one should note that there are two main components to dealing with the noise.
The first is when deciding if the bin is zero, i.e.\ when the only values in the bin and its offsets are noise. 
We initially start willing to assume this is the case and slowly become less willing (by $\delta_z$) to accept that the bin is really zero as we start to try and recover smaller and smaller error rates. 
This means that initially the decoder will concentrate only on bins that have relatively large Pauli errors in them and will be less likely to mistake noise as indicative of an error.

For instance, assuming a reconstruction error normally distributed with a standard deviation of $0.01$, then if a bin contained $2^{14}$ 0-error Paulis with such noise, we would expect the mean of such a bin to be centered around 0, with a standard deviation of $\sqrt(0.01^2/2^{14})\approx7.8\times10^{-5}$. 
Therefore allowing 3 standard deviations we would expect the square of the noise in the bin to be less that $\approx 5.5\times10^{-8}$.  
By ignoring bins with a squared value less than  $5.5\times10^{-8}$ and slowly decreasing this number to, say, $5.5\times10^{-10}$ one can ensure that higher error rate Paulis are first recovered, before exploring possibly empty bins for low error rate Paulis.

The second component is the willingness to accept that there is only one value in the bin offsets. 
This time we start with a strict check, and we only accept a Pauli if the noise is below a threshold, then slowly relax this (by $\delta_s$) as we aim to recover Paulis that happen to have increasing amounts of noise in the bins with them.

\begin{table}[thb]
\begin{algorithm}[H]
	\caption{\label{alg:NoisyDecoder} Noisy Peeling Decoder}
	\begin{algorithmic}[1]
		\Require M $\gets$ Paulis in groups ($C$ sets),\Comment \Cref{fig:circuits}
		\Require $\lambda_{\Psi_{l,c}}$ for $l\in M_c$ $c\in [C]$,\Comment Eigenvalues 
		\Require $\lambda_{(\Psi_{l,c}\oplus b)}$ for $l\in M_c$ $c\in [C]$,$b\in [2^{2n}]$\Comment Offsets
		\Require $\delta_z,\delta_s$\Comment Relaxation parameters (see text)
		\State $Z_s\gets$ initial zero sensitivity
        \State $S_s\gets$ initial singleton sensitivity
        \State $\mathcal{P}\gets$ initialize empty list of Paulis + errors
        \State ~\Comment Set up quasi probability `bins'.
        \For{c=$1,\dots, C$}
		\State  $\tilde{p}_{j,c}\gets \frac{1}{2^n}\sum\limits_{l\in[2^n]}\lambda_{\Psi_{l,c}}(-1)^{\langle j,l\rangle}, j \in\mathbb{F}^n_2$\Comment \Cref{eq:quasiProbl}
		\State  $\tilde{p}_{j,c,b}\gets \frac{1}{2^n}\sum\limits_{l\in[2^n]}\lambda_{(\Psi_{{l,c}}\oplus b)}(-1)^{\langle j,l\rangle}, j \in\mathbb{F}^n_2, b\in 2^{[2n]}$
		\EndFor
		\State ~\Comment Populate $\mathcal{P}$ with singletons.
		\For{$\_=1,2,...,\text{arbitrary}$}
		\For{c=$1,\dots, C$}
		\For{$\tilde{p} \in [\tilde{p}_{j,c},\tilde{p}_{j,c,b}], j\in \mathbb{F}^n_2, b\in 2^{[2n]}$}
		\If{$\text{not IsCloseToZero}(\tilde{p},Z_s)$}
		\If{$\text{IsSingleton}(\tilde{p},S_s)$}
		\State (P,E) $\gets$ SingletonPauliAndSize($\tilde{p}$)
		\State $\mathcal{P}\gets $(P,E)
		\State ~\Comment Remove from other sets
		\State PeelBack(M$_{[C]/c}$,(P,E))
		\EndIf
		\EndIf
		\If{$\sum\mathcal{P}\approx 1$}
		\State \Return $\mathcal{P}$\Comment Success!
		\EndIf
		\EndFor		
		\EndFor
		\If{no new Paulis added since previous iteration}
		\State \Comment Relax search requirements.
		\State $Z_s\gets Z_s-\delta_z$
		\State $S_s\gets S_s+\delta_s$
		\EndIf
		\EndFor
		\State \Return Incomplete $\mathcal{P}$\Comment !Success
	\end{algorithmic}
\end{algorithm}
\end{table}

\section{Provable Recovery Algorithm}\label{sec:algorithm}
In this section, we describe in detail a hashing-based subsampling recovery algorithm for which we prove a recovery guarantee. 
For convenience, we have collected the notation used in the provable recovery algorithm into a glossary of symbols in \Cref{tab:Glossary}.

\begin{table}[th]
    \centering
    \begin{tabular}{cp{0.8\columnwidth}}
        \toprule
        \multicolumn{2}{c}{Glossary of symbols}\\
        \midrule
        $n$ & number of qubits\\
        $N$ & $4^n$, the number of Pauli operators modulo phase\\
        $s$ & sparsity, $s = 4^{\delta n}$ for $0<\delta<\frac12$ \\
        $\delta$ & related to sparsity by $s = 4^{\delta n}$ and $0 < \delta < \frac12$\\
        $B$ & number of bins in a single subsampling group, in the experimental regime described in this paper, typically $B = 2^n$\\
        $b$ & $B=2^b$ for $0<b\leq 2n$. Typically $B=2^n$, but this $b$ is used to indicate the bin number\\
        $\eta$ & the sparsity coefficient, being $B/s$\\
        $C$ & is the number of subsampling groups\\
        $P_1$ & the number of random offsets chosen for each subsampling group\\
        $P_2$ & the number of extra offsets chosen for each subsampling group to form an error correcting code\\
        $P$ & $P=P_1+P_2$\\
        $U_{c,t}[j]$ & the bin generated from subsampling group $c\in[C]$ with index $j\in\mathbb{F}_2^b$, and the subscript $t\in[P]$ indicates the offset this bin uses\\
        $\xi$ & the standard deviation of the noise in the estimated Pauli eigenvalues created by shot noise from the original experiments\\
        $\sigma$ & the standard deviation of the noise in the Pauli error rates created by WHT from the noisy Pauli eigenvalues\\
        $\nu$ & the standard deviation of the noise in a given bin, created by subsampling the noisy Pauli eigenvalues\\
        $\epsilon_0$ & the lower bound on the nonzero Pauli error rates, required by assumption A3\\
        $w_k$ & Noise on the eigenvalue $\lambda_k$, distributed like $\mathcal{N}(0,\xi)$.\\
        $W_m$ & Noise on error rate $p_m$, induced by the WHT.\\
        \bottomrule
    \end{tabular}
    \caption{Glossary of symbols used throughout the proof.}
    \label{tab:Glossary}
\end{table}

Consider a Walsh-Hadamard transformation among $n$-qubit Pauli eigenvalues and Pauli error rates like \eqref{eq:WHT}. 
In order to recover a set of sparse error rates with noisy eigenvalues, it is necessary to consider a noisy variation
\begin{align}\label{eq:noisyWHT}
    \widehat{\lambda}_k=\sum_{m\in\mathbb{F}_2^{2n}}(-1)^{\langle k,m\rangle}p_m+w_k,\ \ k\in\mathbb{F}_2^{2n}.
\end{align}
Note that we employ $w_k$ to indicate the sampling errors of the eigenvalues, and for simplicity we are assuming that they are all independent Gaussian random variables with distribution $\mathcal{N}(0,\xi^2)$.
The proposed algorithm follows the SPRIGHT framework of Ref.~\cite{Li2015}.
It first samples Pauli eigenvalues and forms several groups of bins. 
The algorithm then implements the Peeling Decoder algorithm \ref{alg:peeling} to recover individual Pauli error rates from the sampled bins.

To subsample bins from noisy eigenvalues, this algorithm employs $C$ \textit{subsampling groups}.
Each subsampling group is specified by a binary matrix $\mathbf{M}_c\in\mathbb{F}_2^{2n\times b}$ for $c\in [C]$ and a set of $P$ \textit{offsets}. 
The binary matrices $\mathbf{M}_c$ serve as hash functions to isolate individual Pauli error rates into $B$ bins with high probability. 
In our experimental setting, this $B$ is always chosen as $B=2^n$ for convenience, while in the following proof it's sufficient for $B$ to be as large as $s$, and increasing $B$ exponentially to $s$ is enough to find out the magnitude of $s$. 
Therefore in the proof, we use the general case that $B=O(s)$.
The $P$ offsets provide redundancy designed to make the recovery algorithm robust to limited amounts of sampling noise. 

Before constructing explicit algorithms, we shall introduce a method for choosing offsets by using good error correcting codes~\cite{Li2015}.

\begin{definition}[Offsets]\label{SO}
Let $P=P_1+P_2$ with $P_i={O}(n)$ for $i=1,2$. 
We choose $P_1$ \emph{random offsets} $\mathbf{d}_t$ for $t=0, \cdots,P_1-1$ chosen independently and uniformly at random over $\mathbb{F}_2^n$, and $P_2$ \emph{coded offsets} $\mathbf{d}_t$ for $t=P_1,\cdots,P-1$ such that the offset matrix $\mathbf{G}=[\cdots;\mathbf{d}_t;\cdots;]\in \mathbb{F}_2^{P_2\times 2n}$ constitutes a generator matrix of a linear code with parameters $[P_2,2n,\beta P_2]$ with $\beta > \mathbb{P}$. 
Here $\mathbb{P}$ is an upper bound on the probability that the sample error will change the sign of a single-ton bin (i.e., a bin with a single nonzero Pauli error rate).
\end{definition}

It will be convenient in what follows to define a $2n \times 2n$ matrix $J_n$ given by
\begin{align}
    J_n = X\otimes I_n = \begin{pmatrix}
    0_n & I_n \\
    I_n & 0_n
    \end{pmatrix}\,,
\end{align}
where $0_n$ is the $n \times n$ zero matrix and $I_n$ is the $n \times n$ identity matrix.
This is the symplectic form that controls the commutation relations in the Pauli group. 
That is, if $p,q\in\mathbb{F}_2^{2n}$, then
\begin{equation}
    \langle p,q\rangle = p^T J_n q\,,
\end{equation}
where the arithmetic is implicitly modulo 2.

\begin{table}[t!hb]
 \begin{algorithm}[H] 
  \caption{Subsampling and WHT}\label{alg:subsampling}
  \begin{algorithmic}[1]
    \State \texttt{Input:} Offsets $\mathbf{d}_{c;t}$ for observation index $t \in [P]$ and subsampling index $c \in [C]$; 
    \State \texttt{Input:} Subsampling matrices $\mathbf{M}_c\in\mathbb{F}_2^{2n\times b}$ for some $b>0$ and $c \in [C]$. 
    \State \texttt{Modify:} $\mathbf{M}_c'\leftarrow J_n \mathbf{M}_c J_b$ $\forall\,c\in[C]$.
    \ForAll{ $c\in[C]$, $t\in[P]$, and $\ell\in\mathbb{F}_2^{b}$}
        \State $k \leftarrow \mathbf{M}_c' \ell+\mathbf{d}_{c;t}$
        \State ${\tt Estimate}:$ $\widehat{\lambda}_k$ 
    \EndFor
    \State $B \leftarrow 2^{b}$
    \ForAll{ $c\in[C]$ and $t\in[P]$}
    	\State \label{algline:W} 
	$U_{c;t}[j] \leftarrow\frac{1}{B} \sum_{\ell\in\mathbb{F}_2^{b}} (-1)^{\langle j,\ell\rangle}\widehat{\lambda}_{\mathbf{M}_c' \ell+\mathbf{d}_{c;t}}$
	\State Return $U_{c;t}[j]$
	\EndFor	
  \end{algorithmic}
\end{algorithm}
\end{table}

We now introduce \Cref{alg:subsampling} to use for data preprocessing and bin construction.
The indices on each array are considered to be modulo their respective dimension, and each element of the summation $\mathbf{M}_c' \ell+\mathbf{d}_{c;t}$ is calculated in the field $\mathbb{F}_2$. 
The algorithm calculates bin coefficients using the corresponding binary matrices and by taking sums over the whole space $\mathbb{F}_2^{b}$. 
After this subsampling process, each subsampling group contains $P$ sets of $B=2^{b}$ bins, where $b$ is a free parameter. 
The result of applying \Cref{alg:subsampling} is summarized in the following lemma. 

\begin{lemma}[\bf Basic Observation Model]\label{lm:prop_hashing_obs}
The $B$-point WHT subsampled bin coefficients with index $j\in\mathbb{F}_2^{b}$ can be written as:
\begin{align}\label{eq:U_cp}
	{U}_{c;t}[j] &= \sum_{m:\,\mathbf{M}_c^Tm = j} p_m(-1)^{\langle\mathbf{d}_{c;t},m\rangle}+{W}_{c;t}[j], \,\forall t\in[P].
\end{align}
 Moreover the sample error is as follows
\begin{align*}
    {W}_{c;t}[j] &= \sum_{m:\,\mathbf{M}_c^Tm = j} {W}_m(-1)^{\langle\mathbf{d}_{c;t},m\rangle},
\end{align*}
where ${W}_m$ is the noise of the Pauli error rate $p_m$. 
\end{lemma}

We remark that the noise $W_m$ on the error rate $p_m$ is induced by the Gaussian noise $\{w_k\}$ on the noisy eigenvalues $\{\lambda_k\}$ by the WHT. 
It is important to keep these two noise sources separate, although we will not make much direct use of $w_k$ in the remainder of the paper. 

\begin{proof}
Denote $p_m+W_m$ by $\tilde{p}_m$, so the noisy Pauli eigenvalues can be transformed to 
\begin{gather*}
    \widehat{\lambda}_k=\sum_{m\in\mathbb{F}_2^{2n}}(-1)^{\langle k,m\rangle}\tilde{p}_m,\ \forall\,k\in\mathbb{F}_2^{2n}.
\end{gather*}
From \Cref{alg:subsampling}, a specific bin $U_{c;t}[j]$ for some $c\in[C]$, $t\in[P]$, and $j\in\mathbb{F}_2^{2n}$ is constructed as follows, 
\begin{align*}
    U_{c;t}[j]=\frac{1}{B}\sum_{\ell\in\mathbb{F}_2^{b}}\sum_{m\in\mathbb{F}_2^{2n}}(-1)^{\langle m,\mathbf{M}'_c\ell\rangle}(-1)^{\langle j,\ell\rangle}(-1)^{\langle m,\mathbf{d}_{c;t}\rangle}\tilde{p}_m.
\end{align*}
A key observation is 
\begin{align*}
    \langle m,\mathbf{M}'\ell\rangle
    =&m^TJ_n\mathbf{M}'\ell\\
    =&m^TJ_n\cdot J_n\mathbf{M}J_b\ell\\
    =&(\mathbf{M}^Tm)^TJ_b\ell\\
    =&\langle \mathbf{M}^Tm,\ell\rangle,
\end{align*}
where the first equation is from the definition of the symplectic inner product in~\cref{sec:prelim}, and the second equation comes from {\tt Modify} part in \Cref{alg:subsampling}, and the third is due to the following property of $J$
\begin{gather}
    J_n\cdot J_n = I_{2n}\ \forall\,n\in\mathbb{N}.
\end{gather}
Thus the bin can be simplified as follows,
\begin{align*}
U_{c;t}[j]&=\frac{1}{B}\sum_{m\in\mathbb{F}_2^{2n}}\sum_{\ell\in\mathbb{F}_2^{b}}(-1)^{\langle\mathbf{M}_c^Tm+j,\ell\rangle}(-1)^{\langle m,\mathbf{d}_{c;t}\rangle}\tilde{p}_m \\
&=\sum_{m:\,\mathbf{M}_c^Tm = j}\tilde{p}_m(-1)^{\langle m,\mathbf{d}_{c;t}\rangle}\\
&= \sum_{m:\,\mathbf{M}_c^Tm = j} p_m(-1)^{\langle\mathbf{d}_{c;t},m\rangle}+{W}_{c;t}[j],
\end{align*}
where ${W}_{c,t}[j]$ is defined in the lemma.
\end{proof}
We note that from Line~\ref{algline:W} in \Cref{alg:subsampling}, the fact that the original noise $\boldsymbol{w}$ is isotropic, and the fact that the $\ell$-bit WHT is proportional to an orthogonal transformation, it follows that the noise in each bin $\boldsymbol{W}_{c}[j]$ remains Gaussian distributed, but according to the distribution $\mathcal{N}(0,\nu^2\mathbbm{1})$ where $\nu^2=\frac{\xi^2}{B}$. 
Moreover, we can combine each $U_{c;t}[j]$ for different $t\in[P]$, and get a vector
\begin{gather*}
    \mathbf{U}_c[j]:=[U_{c;0}[j],\cdots,U_{c;P-1}[j]]^T,
\end{gather*}
and an analogous vectorization can be implemented on the offsets
\begin{gather*}
    \mathbf{D}_c:=[\mathbf{d}_{c;0}\cdots\mathbf{d}_{c;P-1}].
\end{gather*}
Therefore, \Cref{lm:prop_hashing_obs} can be rewritten as follows.

\begin{lemma}[\bf Bin Observation Model]\label{lm:rewritten_prop_hashing_obs}
The $B$-point WHT subsampled bin with index $j\in\mathbb{F}^{b}$ in the $c$-th subsampling group is
\begin{gather}\label{eq:bin_description}
    \mathbf{U}_{c}[j] = \sum_{m:\,\mathbf{M}_c^Tm = j} p_m(-1)^{\langle\mathbf{D}_c,m\rangle}+\mathbf{W}_{c}[j],
\end{gather}
where the noise $\mathbf{W}_c[j]=\sum_{m:\,\mathbf{M}_c^Tm=j}W_m(-1)^{\langle\mathbf{D}_c,m\rangle}$ is distributed as $\boldsymbol{W}_{c}[j]\sim\mathcal{N}(0,\nu^2\mathbbm{1})$ with $\nu^2=\frac{\xi^2}{B}$, and ${W}_m$ is the WHT noise of Pauli error rate $p_m$.
\end{lemma}

\begin{proof}
This is a variation of \Cref{lm:prop_hashing_obs}.
\end{proof}

After subsampling and calculating bins, it's straightforward to design a protocol to extract information from these bins. 
The idea is to construct a bipartite graph $G$, as in \Cref{fig:bipartite}, with $s$ left nodes representing nonzero Pauli error rates and $BC$ right nodes representing bin vectors $\mathbf{U}_c[j]$.

We draw an edge from each left node (a nonzero Pauli error rate) to every right node that contains that Pauli. 
Each Pauli error rate will occur exactly once in each subsampling group, the degree of the left nodes is therefore $C$. 
We can use the resulting degrees of the right hand nodes to partition them into three types.
We call a bin with exactly one nonzero Pauli error rate a \textit{single-ton}, and similarly there are \textit{zero-ton} and \textit{multi-ton} bins that contain exactly zero or contain more than one Pauli error rate respectively. 
(Recall that this graph is depicted in \Cref{fig:bipartite}.)
Shortly we will describe in detail a method to detect which type of bin a particular node has been partitioned into.
After invoking such a bin detector, the peeling decoder can be designed to peel out the detected single-ton Pauli error rates by subtracting them from every multi-ton bin in which they appear, removing the associated edge from the graph.
This will reduce the degree of that right-hand node, potentially turning it from a multi-ton bin into a single-ton bin.
For the range of parameters that we have chosen and the assumptions outlined above, iterating this decoder to discover new single-tons and reduce multi-tons will converge to reduce the graph to only zero-ton and single-ton bins with high probability. 

In the \Cref{alg:peeling}, we apply an array $\mathbf{T}$ that indicates the variance of the propagated noise part, $W_{c;t}[j]$, in each bin. 
These numbers help track the propagation of error in the bin detector from the calculation in Line~\ref{algline:U} of \Cref{alg:peeling}. 
The equation in Line~\ref{algline:T} of that algorithm describes how to update $\mathbf{T}$. 
\Cref{lm:variance} below shows the need and utility of this parameter.

One subtlety to applying the peeling decoder to this graph is that the graph might have cycles. 
Peeling on a graph with cycles will in general lead to dependencies in the random variables, which complicates the analysis. 
However, as we show below in \Cref{lm:tree-like}, large local neighborhoods of the peeling graph look locally tree-like with high probability, therefore we can peel for a large number of steps before encountering a cycle. 
With the correct choice of parameters, the tree-like neighborhood can be made large enough throughout the graph to ensure convergence of the peeling decoder. 

As we previously mentioned, the peeling decoder algorithm is based on a subroutine that we call \textit{Bin Detector} (it is set out in \Cref{alg:bin_detect}). 
We will denote it by $\mathrm{BD}(\mathbf{U}_c,\mathbf{D}_c,T)$. 
The subroutine, $\mathrm{BD}$, will take a bin, the offsets chosen, and a noise parameter $T$ as inputs, and it will output an estimate for the type (zero-ton, single-ton, or multi-ton) of the bin $\widehat{\mathfrak{B}}$, and if the bin is a single-ton it also returns the estimated index $\widehat{m}$ and Pauli error rate $\widehat{p}_{\widehat{m}}$.
The subroutine $\mathrm{BD}$ also depends on two parameters $\gamma_1$ and $\gamma_2$, but these can be chosen as arbitrary constants in the interval $(0,1)$. 
Their only purpose is to ensure exponential decay of the failure probability of bin detection, as we discuss in \Cref{lm:failure_bound}. 

\begin{table}[thb]
\begin{algorithm}[H] 
  \caption{Peeling Decoder\label{alg:peeling}}
  \begin{algorithmic}[1]
    \State ${\tt Input}:$ observation vectors
    $\mathbf{U}_c[\emph{j}]$, offsets $\mathbf{D}_c$ and array $\mathbf{T}_{c}[\emph{j}]$ initialized by 1 for $\emph{j} \in \mathbb{F}_2^{b}$, $c\in[C]$; 
    \State ${\tt Input}:$ the number of peeling iterations $I$;  
    \State $\mathcal{P}\gets$ initialize empty list of Paulis ($\widehat{m},\widehat{{p}}_{\widehat{m}}$)
    \For{$i\in[I]$}
    	\ForAll{$c\in[C]$ and $\emph{j}\in\mathbb{F}_2^{b}$}
			\State $(\widehat{\mathfrak{B}},\widehat{m},\widehat{{p}}_{\widehat{m}}) \leftarrow \mathrm{BD}(\mathbf{U}_c[\emph{j}],\mathbf{D}_c,\mathbf{T}_{c}[\emph{j}])$
			\If{$\widehat{\mathfrak{B}}= \textrm{single-ton}$}
			\State $\mathcal{P}\gets$ ($\widehat{m},\widehat{{p}}_{\widehat{m}}$)
			\ForAll{$c'\in[C]$ and $c'\neq c$}
			    \State ${\tt Locate~bin~index}$ $\emph{j}_{c'} \leftarrow \mathbf{M}_{c'}^T \widehat{m}$
			    \State \label{algline:T} $\mathbf{T}_{c'}[j_{c'}]\leftarrow\mathbf{T}_{c'}[j_{c'}]+\frac{\mathbf{T}_c[j]}{P_1}+\frac{(P_1-1)B}{P_1N}$
			        \State \label{algline:U} $\mathbf{U}_{c'}[\emph{j}_{c'}] \leftarrow \mathbf{U}_{c'}[\emph{j}_{c'}] - \widehat{p}_{\widehat{m}}(-1)^{\langle\mathbf{D}_{c'},\widehat{m}\rangle}$
			 \EndFor
			\ElsIf {${\widehat{\mathfrak{B}}}\neq \textrm{single-ton}$}
				 \State continue to next $\emph{j}$.
			\EndIf			
		\EndFor
    \EndFor
    \State Return: $\mathcal{P}$
  \end{algorithmic}
\end{algorithm}
\end{table}

By using the sparsity assumption and our choice of subsampling matrices, this peeling process will succeed with high probability. 
Intuitively we can see this from our ability to choose subsampling matrices $\mathbf{M}$ in such a way that we can find bins that typically contain only zero or one nonzero Pauli error rate. 
In \cite{Li2015} the authors provide a proof that if a bin detector algorithm always returns an exactly correct answer, then the oracle-based peeling decoder has a failure probability that vanishes with the signal size. 
So it suffices to propose a suitable design for the bin detector and a corresponding recovery guarantee.

In designing such a bin detector we will need to estimate the index of the relevant Pauli in the bin. 
For index estimation in the setting where there is noise we will need to make our estimation robust. 
One approach is to use some repetition of detection and a majority voting. 
A better approach is to use some form of error correcting code for the offsets, as discussed above in \Cref{SO}. 
In what follows, we will use the following definition of a sign function,
\begin{equation}
    \sgn{x} = \begin{cases}
    0 & \text{if } x \ge 0,\\
    1 & \text{if } x < 0.
    \end{cases}
\end{equation}
With this definition, we have the following lemma, which confirms that offsets chosen in accordance with \Cref{SO} can be used to estimate the indices.

\begin{lemma}\label{lm:MLE}
Given a single-ton bin $(m,p_m)$ observed with noise
\begin{align}\label{single-ton}
	U &= p_m(-1)^{\langle\mathbf{d},m\rangle} + W,
\end{align}
and supposing that the variance in each row (offset, $\mathbf{d}$)  of the bin is equal to $T\nu^2$, then the sign of each observation satisfies
\begin{align}\label{sign}
	\sgn{U}
	&= \langle\mathbf{d},m\rangle\oplus Z,
\end{align}
where $Z$ is a Bernoulli random variable with probability $\Pr(Z=1) \leq \mathbb{P}_m \colonequals \sqrt{\frac{T\nu^2}{2\pi p^2_m}}\mathrm{e}^{-\frac{p_m^2}{2T\nu^2}}$.
\end{lemma}

\begin{proof}
The first term in \eqref{sign} follows trivially from the sign of the power of minus one in \eqref{single-ton} and the fact that $p_m$, being a probability, is always positive. 
The second term, $Z$, will be 1 if and only if $|W|$ is larger than $p_m$ so that it can change the sign generated by the first term of \eqref{single-ton}. 
Therefore, $Z$ is Bernoulli distributed with a probability that we can bound as follows. 
Recalling that $W$ is a Gaussian random variable, we can use the relevant tail bounds for our assumption on the variance (for details see~\cite{feller2008introduction}) to obtain
\begin{align}
    \Pr(Z = 1) &= \frac{1}{2}\Pr(|W|>p_m)=\Pr(W>p_m)\notag\\
    &\leq \sqrt{\frac{T\nu^2}{2\pi p^2_m}}\mathrm{e}^{-\frac{p_m^2}{2T\nu^2}}=\mathbb{P}_m,
\end{align}
where $T$ is the number extracted from the array $\mathbf{T}$ in \Cref{alg:peeling}. 
\end{proof}
\begin{remark}\label{strictly_small}
\rm If we assume that the maximum degree of right nodes in the bipartite graph $G$ is not larger than $\tfrac{1}{2}P_1$, $\mathbb{P}_m<\frac{1}{2}$ is satisfied for all $m\in\mathbb{F}_2^{2n}$ (using \textbf{A1} in \Cref{random_support_assumption} and \Cref{maximum_T}). 
We will bound the probability that this right-hand degree assumption fails in \Cref{lem:Bounded_degree}.
\end{remark}

In what follows, we will ignore the subscripts $c$ and indices $j$ of the bins when it will not lead to any misunderstanding. 

Now \Cref{lm:MLE} can be used to identify $\widehat{m}$, the index of the Pauli error rate in a single-ton bin. 
Given the offsets chosen in \Cref{SO} and recalling \Cref{lm:rewritten_prop_hashing_obs}, we have the following equation from the code generator $\mathbf{G}$ for the signs of every element in a bin,

\begin{align}
    \begin{bmatrix}
		\sgn{U_{P_1}}\\
		\vdots\\
		\sgn{U_{P-1}}
	\end{bmatrix}
	&=
	\langle\mathbf{G},m\rangle
	\oplus
	\begin{bmatrix}
		Z_{P_1}\\
		\vdots\\
		Z_{P-1}
	\end{bmatrix}.\label{eq:sign_code}
\end{align}
Since the bit length of the index $m$ is $2n$, we can choose the number $P_2$ as follows. 
We choose any linear code with rate $R$ and distance $d$, and a decoder that can decode up to at least a minimum distance $\beta 2n / R$ for parameters $\beta, R = \Theta(1)$. 
Obviously this requires $d \ge \beta 2n /R$. 
The additional constraint on $\beta$ is that $\beta$ is larger than the probability $\mathbb{P}$ of any of the Bernoulli random variables $Z_i$ to be 1. 
Then we can choose $P_2=2n/R$.
That is, we are looking for a classical linear code with parameters $[2n/R, 2n, d \ge \beta 2n / R]$.
There are a number of pre-existing candidate codes that can be decoded up to a constant fraction of the minimal distance in linear time in the length of the code exist that satisfy these stringent conditions. 
For example, expander codes~\cite{sipser1996expander} can be implemented to construct the code generator $\mathbf{G}$ and the parity check matrix $\mathbf{H}$, and the greedy linear-time decoder~\cite{sipser1996expander} can correct errors with weight up to $d/4$.
The decoder of the corresponding code is required to retrieve the estimate $\hat{m}$.
Since the manner of coding and decoding is flexible, here we only use $\texttt{Decode}$ to indicate the decoder.
\begin{gather}
    \widehat{m}=\texttt{Decode}\left( \begin{bmatrix}
		\sgn{U_{P_1}}\\
		\vdots\\
		\sgn{U_{P-1}}
	\end{bmatrix}\right).
\end{gather}

With this we can specify the modified algorithm to detect the bin $U$ with the offsets as in \Cref{SO} along with the corresponding number $T$.

We are now ready to give a precise specification of the bin detector algorithm. 
\begin{table}[thb]
  \begin{algorithm}[H]
  \caption{Bin Detector: $\mathrm{BD}(\mathbf{U}_c,\mathbf{D}_c,T)$}\label{alg:bin_detect} 
  \begin{algorithmic}[1]
  \State {\tt Input}: bin $\mathbf{U}_c$, offsets $\mathbf{D}_c$ and the number $T$ to indicate error size;
  \State\label{algline:zero-ton} {\tt Parameter}: real numbers\label{algline:parameter} $\gamma_1,\gamma_2\in(0,1)$;
  \If{$\frac{1}{P_1}\sum_{t=0}^{P_1-1}U_{c;t}^2\leq T(1+\gamma_1)\nu^2$}
    \State $\widehat{\mathfrak{B}} \leftarrow $\texttt{zero-ton}
    \State Return ($\widehat{\mathfrak{B}}$, \textit{nil, nil})    \Comment \textit{zero-ton verification}
  \EndIf
  \State $\widehat{m}\leftarrow \texttt{Decode}([\sgn {U_{P_1}},\cdots,\sgn{U_{P-1}}]^T)$
  \State $\widehat{p}_{\widehat{m}}\leftarrow\frac{1}{P_1}\sum_{t=0}^{P_1-1}(-1)^{\langle\mathbf{d}_{c;t},\widehat{m}\rangle}U_{c;t}$ \Comment \textit{single-ton search}
  \If{$\frac{1}{P_1}\sum_{t=0}^{P_1-1}(U_{c;t}-(-1)^{\langle\mathbf{d}_{c;t},\widehat{m}\rangle}\widehat{p}_{\widehat{m}})^2\leq T(1+\gamma_2)\nu^2$}
    \State $\widehat{\mathfrak{B}} \leftarrow$ \texttt{single-ton}
    \State\label{algline:single-ton_verification} Return ($\widehat{\mathfrak{B}}$, $\widehat{m}$, $\widehat{p}_{\widehat{m}}$) \Comment \textit{single-ton verification} 
  \Else 
  \State $\widehat{\mathfrak{B}} \leftarrow$ \texttt{multi-ton}
  \State Return ($\widehat{\mathfrak{B}}$,\textit{nil, nil}) 
  \EndIf
  \end{algorithmic}
\end{algorithm}
\end{table}

\section{Proof of Main Theorem}\label{sec:proof}

We now repeat the statement of the main theorem. 
\begin{thmbis}{Thm:main}
Suppose the Assumptions \ref{random_support_assumption} hold for an unknown Pauli channel with eigenvalues $\boldsymbol{\lambda}$ and error rates $\boldsymbol{p}$. 
Then with failure probability $\mathbb{P}_F\leq \mathrm{e}^{-O(n)}$, Algorithms~\ref{alg:subsampling}, \ref{alg:peeling}, and \ref{alg:bin_detect} estimate the $s$-sparse Pauli error rates $\widehat{\boldsymbol{p}}$ such that $\|\widehat{\boldsymbol{p}}-\boldsymbol{p}\|_{\infty}\leq 2\xi/\sqrt{B}$ using $O\bigl(s n\bigr)$ eigenvalue queries and $O(sn^2)$-time classical computation.
\end{thmbis}

Recall that using the protocol in \cite{Flammia2019}, we can estimate $O(s n)$ eigenvalues to within variance $\xi^2$ by doing $O(\frac{n^2}{\xi^2})$ measurements. 
Therefore, heuristically, the entire algorithm needs $O(\frac{n^2}{\xi^2})$ measurements to achieve this recovery guarantee. 

\begin{proof}
Firstly we consider the stated query and computational complexities. 
From \cite[Theorem 2]{Li2015}, it's shown that the oracle-based peeling decoder succeeds with probability $1-O(1/s)$ for a random sparse set (obeying assumption \textbf{A1}) as long as $C=O(1)$ and $B=O(s)$. 
Therefore, to prepare these bins, the number of queried eigenvalues is $BPC=O(Ps)$. 

To construct the bins and the corresponding graph, the computational complexity can be calculated by the complexity from the construction algorithm. 
Note there are $P$ offset coefficients $\mathbf{d}$ and each $\mathbf{U}_{c,t}[j]$ comes from the sum of $B$ samples in \Cref{alg:subsampling}. 
To construct the total set $\{\mathbf{U}_{c,t}[j]\}_{C,P,B}$, the we can use a fast WHT (which has complexity $O(B\log B)$ to calculate a $B$-point WHT) for each offset. Therefore, the computational complexity for this part is
\begin{gather*}
    Z_1=O(PB\log B)=O(Psn)\,. 
\end{gather*}
The second part of computational complexity comes from the computation of \Cref{alg:peeling}. 
Each step in the bin detector checks the type of the bin with $O(P)$ calculations, and there are  $O(B)$ iterations. 
Accordingly, the complexity is 
\begin{gather*}
    Z_2=O(PB)=O(Ps)\,.
\end{gather*}
Therefore, the total computational complexity is $Z=Z_1+Z_2=O(Psn)$.

Let us denote by $E_{\text{bin}}$ the event that any invocation of the bin detector (in the execution of \Cref{alg:peeling}) returns one or more of the following: (a) an incorrect identification of the type of bin; (b) wrong indices for a detected single-ton, or (c) a mis-estimate of the Pauli error rates of a detected single-ton by more than  $2\nu=2\xi/\sqrt{B}$. 
Furthermore, let $D$ denote the event that the maximum degree of the right nodes in the graph $G$ is less than or equal to $P_1$. 
Let $H$ be the event that all the peeling routes in the procedure are cycle-free. 
Then utilizing the law of total probability we can bound the failure rate of the entire algorithm as
\begin{align}
    \mathbb{P}_F 
	\leq &\Prob{\text{Peeling decoder fails}\big|{E}_{\rm bin}^c}\notag\\
	&+\Prob{E_{\rm bin}|D,H}+\Prob{D^c}+\Prob{H^c} \label{error_rate}.
\end{align}
Here the subscript $c$ denotes the complement of the event, e.g.~${E}_{\rm bin}^c$, denotes that no bin detection error occurred in the entire execution of \Cref{alg:peeling}.

The first term in \eqref{error_rate} is the chance that the oracle-based peeling decoder fails, even though the bin decoder is always correct. This probability scales as  $O(1/s)$ (Proposition 4 in \cite{Li2015}). 

To bound the second term, it will be more convenient to consider the probability that every invocation of a bin detector works correctly given $D$ and $H$. 
Let $M$ denote the number of times the peeling decoder calls the  bin detector subroutine. This probability can be expressed as follows,
\begin{align*}
    &\Prob{\bigcap_{i=1}^M E_i^c\bigg|D,H}\\
    &=\Prob{E_M^c\bigg|\bigcap_{i=1}^{M-1}E_i^c,D,H}\Prob{\bigcap_{i=1}^{M-1}E_i^c\bigg|D,H}\\
    &=\Prob{E_M^c\bigg|\bigcap_{i=1}^{M-1}E_i^c,D,H}\cdots\Prob{E_1^c|D,H},
\end{align*}
where $E_i$ denotes the event that the $i$th call of bin detector returns a wrong answer.
According to \Cref{lm:variance}, the parameter $T$ will always correctly estimate the  variance if all the earlier bin detectors worked correctly. 
From \Cref{lm:failure_bound}, each term in the above equation will be lower bounded by
\begin{gather*}
    \Prob{E^c|D,V,H}\geq 1-\mathrm{e}^{-O(P_1)},
\end{gather*}
where $V$ here just indicates that all the previous bin detectors work correctly.
So we have 
\begin{align*}
    \Prob{\bigcap_{i=1}^M E_i^c\bigg|D,H}&\geq\left(1-\mathrm{e}^{-O(P_1)}\right)^M\,.
\end{align*}
Moreover, since $M$, the number of times the bin detector routine is called, is at most $O(BCs)$, the upper bound of the second term is
\begin{align*}
    \Prob{{E}_{\rm bin}|D,H} &\leq 1 - \left(1-\mathrm{e}^{-O(P_1)}\right)^M \\
    &\leq O(BCs)\mathrm{e}^{-O(n)} \leq \mathrm{e}^{-O(n)}\,.
\end{align*}

\Cref{lem:Bounded_degree} provides that the third term in \eqref{error_rate} is also exponentially decaying with $P_1$:
\begin{gather*}
    \Prob{D^c}\leq \mathrm{e}^{-O(P_1)}\leq\mathrm{e}^{-O(n)},
\end{gather*}
where the last inequality comes from the definition of $P_1$ from \Cref{SO}.
Similarly \Cref{lm:tree-like} and Remark~\ref{iteration number} provide the bound on the probability of $H^c$:
\begin{align*}
    \Prob{H^c}\leq O\left(\frac{\log^{\log\log (s)} s}{s}\right) \leq \mathrm{e}^{-O(n)}\,.
\end{align*}
Therefore, the total failure probability of our peeling decoder algorithm is vanishing exponentially with the number of qubits $n$. 
And from \Cref{SO}, the total number of offsets consists of $P_1=\Theta(n)$ random offsets and $P_2=\Theta(n)$ coding offsets, thus $P=\Theta(n)$ and stated complexities have been proven.
\end{proof}

\section{Tail bounds\label{sec:tailbounds}}

In this section, we prove several statements bounding the failure probabilities of various events that can cause the bin detector outlined as \Cref{alg:bin_detect} to fail. 
One of the main lemmas that we will need is the following tail bound on Gaussian random variables.

\begin{lemma}[{Tail bound~\cite[Lemma 11]{Li2015}}]\label{lem:tailbound}
Given $\textbf{g},\textbf{k} \in \mathbb{R}^N$ where $\textbf{k}$ is an isotropic Gaussian random variable $\textbf{k}\sim\mathcal{N}(0,\nu^2\mathbbm{1}_N)$, then the following tail bound holds:
\begin{align}
    \Pr\left(\frac{1}{N}\|\textbf{g}+\textbf{k}\|^2\geq\tau_1\right)\leq&\mathrm{e}^{-\frac{N}{4}\left(\sqrt{2\tau_1/\nu^2-1}-\sqrt{1+2\theta_0}\right)^2}\label{tailbound1}\\
    \Pr\left(\frac{1}{N}\|\textbf{g}+\textbf{k}\|^2\leq\tau_2\right)\leq&\mathrm{e}^{-\frac{N}{4}\frac{\left(1+\theta_0-\tau_2/\nu^2\right)^2}{1+2\theta_0}}\label{tailbound2},
\end{align}
for $\tau_1, \tau_2$, and $\theta_0$ satisfying
\begin{gather*}
    \tau_1\geq\nu^2(1+\theta_0),\ \tau_2\leq\nu^2(1+\theta_0),\ \theta_0 = \frac{\|\textbf{g}\|^2}{N\nu^2}.
\end{gather*}
\end{lemma}

Since we will use this \Cref{lem:tailbound} to get a failure bound, it is critical to show the sample errors within the different offsets for a particular bin are independent. 
However, given that bins are created as shown in Line~\ref{algline:U} in \Cref{alg:peeling}, it is not immediately clear that the sample errors remain independent. 
To show this independence let us first extend the definition of the sample errors $W_{c;t}[j]$ to take into account the effect of the peeling decoder \Cref{alg:peeling}. 
Recall that for any particular bin we have $P=P_1+P_2$ offset bins.

\begin{definition}\label{def:error_part}
For a specific bin, regard $\mathbf{U}_c[j]$ as a vector of length $P$ as in \Cref{lm:rewritten_prop_hashing_obs}. 
At a given time step in \Cref{alg:peeling}, denote the set of indices of the current contained non-zero Pauli error rates by $\mathcal{P}$.
Define the \emph{random offset errors} and the \emph{coded offset errors} by the equation
\begin{gather*}
    W_{c;t}[j]:=U_{c;t}[j]-\sum_{m\in\mathcal{P}}(-1)^{\langle\mathbf{d}_{c;t},m\rangle}p_m.
\end{gather*}
where $t \in [P_1]$ for the random offset errors and $t \in P_1 + [P_2]$ for the coded offset errors.
We can combine all of these sample errors to define $\mathbf{W}_c[j]$ following the manner of \Cref{lm:rewritten_prop_hashing_obs}.
\end{definition}

In order to discuss the independence of errors and the evolution of their variance, we first introduce some results to rule out some delicate situations.
To define this more rigorously, consider the \textit{directed neighborhood} $\mathcal{N}_e^l$ in the bipartite graph which consists of nonzero Pauli error rates (left nodes) and all the bins (right nodes). 
The neighborhood $\mathcal{N}_e^{2l}$ with length $2l$ and an edge $e=(v,c)$ is an induced subgraph containing all the edges and nodes on paths $e_1,e_2\cdots,e_{2l}$ from node $v$ where $e\neq e_1$.

Denote by $\mathcal{T}_l$ the event that for every edge in the bipartite graph, this subgraph $\mathcal{N}^{2l}_e$ is cycle-free. 
If $\mathcal{T}_l$ occurs, then all the first $l$ peeling iterations will progress independently and there will be no initial error propagating to any bin more than once in the first $l$ iterations.
It has been shown in \cite{Li2015} that with sufficiently large $s$ and $N$, the effective part of the subsampling-based bipartite graph, similar to \Cref{fig:bipartite}, is sufficiently cycle-free for our purposes, as the following lemma illustrates.
\begin{lemma}[{Ref.~\cite[Lemma 6]{Li2015}}]\label{lm:tree-like}
For any iteration $l$, the probability of the complement of $\mathcal{T}_l$ is bounded as
\begin{align*}
    \Pr{(\mathcal{T}^c_l)}\leq c_0\cdot \frac{\log^l s}{s},
\end{align*}
for some constant $c_0$.
\end{lemma}
\begin{remark}\label{iteration number}
\rm Ref.~\cite{Li2015} shows that the probability $p_l$ that an arbitrary edge remains after $l$ peelings given that given that the event $\mathcal{T}_l$ is true can be calculated recursively as
\begin{align}
    p_l=\begin{cases}
    \ \ \ \ \ \ \ 1&  \text{$l$=0}\\
    \left(1-\mathrm{e}^{-p_{l-1}/\eta}\right)^{C-1}&  \text{otherwise}
    \end{cases},
\end{align}
where $\eta$ is the factor $\frac{B}{s}$ and $C$ is the number of subsampling groups. 
\end{remark}

To illustrate the convergence of \Cref{alg:peeling} given the event $\mathcal{T}_l$, consider taking $C=6$ and $\eta=1$, which are reasonable choices in the regime of interest. 
Then the probability of an edge will decrease to approximately machine precision in only three iterations. 
In general this $p_l$ vanishes exponentially with an exponent of $l$. 
Using the law of total probability, the probability of that there exist any edges after $l= \Omega(\log\log s)$ iterations is
\begin{align}\label{halt_probability}
  p_l + \Pr(\mathcal{T}_l^c) = O\left(\frac{\log^{\log\log s} s}{s}\right).
\end{align}
Therefore, the event that there exist bins getting peeled by some earlier bins in a cyclic manner during the whole process happens with probability of the same magnitude of~\eqref{halt_probability}, which converges to zero with $s$.

\begin{lemma}\label{lm:independence}
For an arbitrary timestamp in \Cref{alg:peeling}, sample errors  in each of the random offset errors for a particular bin $\mathbf{U}_c[j]$ remain independent of each other given that the peeling route is cycle-free.
\end{lemma}
\begin{proof}
For an initial bin subsampled from \Cref{alg:subsampling}, consider an arbitrary pair of offsets labeled by $t_1, t_2\in[P_1]$ in the same bin $\mathbf{U}_c[j]$
\begin{align*}
    W_{c;t_1}[j] &= \sum_{m:\,\mathbf{M}_c^Tm = j} {W}_m(-1)^{\langle\mathbf{d}_{c;t_1},m\rangle},\\
    W_{c;t_2}[j] &= \sum_{m:\,\mathbf{M}_c^Tm = j} {W}_m(-1)^{\langle\mathbf{d}_{c;t_2},m\rangle}.
\end{align*}
Since all the errors $W_m$ are i.i.d.\ Gaussian random variables $\mathcal{N}(0,\frac{\xi^2}{N})$, it is obvious that $\mathbb{E}(W_{c;t_1}[j]\cdot W_{c;t_2}[j])=0$. 
So they are independent given that the expected values of the samples errors are 0.

The peeling decoder in Line~\ref{algline:U} in \Cref{alg:peeling} causes errors in the estimate of $W_{c;t}[j]$ to propagate in the following manner
\begin{gather*}
    W_{c;t}[j]\leftarrow W_{c;t}[j]+(p_m-\widehat{p}_{\widehat{m}})(-1)^{\langle\mathbf{d}_{c;t},\widehat{m}\rangle}.
\end{gather*}
We now proceed by induction. 
As discussed above, the noise is initially independent, so the base case is satisfied. 
Now assume that the sample errors before peeling are independent of each other. 
Observing that the updated error still has mean zero, we can calculate the expected value of a product between an arbitrary pair of sample errors in the offsets of a bin to show independence between the offset bins. 
For convenience, denote the updated error by $W_{c;t}[j]'$.
Then we have
\begin{align*}
    \mathbb{E}(W_{c;t_1}[j]'&\cdot W_{c;t_2}[j]') =\\
    \mathbb{E}\bigl(&W_{c;t_1}[j]\cdot W_{c;t_2}[j] \\ 
    &+W_{c;t_1}[j]\cdot(p_m-\widehat{p}_{\widehat{m}})(-1)^{\langle\mathbf{d}_{c;t_2},\widehat{m}\rangle}\\
    &+W_{c;t_2}[j]\cdot(p_m-\widehat{p}_{\widehat{m}})(-1)^{\langle\mathbf{d}_{c;t_1},\widehat{m}\rangle}\\
    &+(p_m-\widehat{p}_{\widehat{m}})^2(-1)^{\langle\mathbf{d}_{c;t_2}+\mathbf{d}_{c;t_1},\widehat{m}\rangle}\bigr)=0.
\end{align*}
The first three terms vanish because the noise has zero mean and the peeling route is cycle-free, and the last term vanishes because the expectation over the independent random offset phases is $\mathbb{E}\bigl((-1)^{\langle\mathbf{d}_{c;t_i},m\rangle}\bigr) = 0$ for any $t_i \in [P_1]$ and $m$. 
\end{proof}

In \Cref{alg:peeling}, we employ an array $\mathbf{T}$ to keep track of the variance of sample error for each bin. 
This array gets updated whenever the algorithm peels a bin using an estimated Pauli error rate. 
We now show that this does indeed correctly track the variance of the sample errors in the bins.

\begin{lemma}\label{lm:variance}
Suppose that at a given arbitrary timestamp in \Cref{alg:peeling} all the bin detector subroutines called earlier have correctly identified their bins and the peeling route is cycle-free.
Then for each bin and its corresponding offsets $\mathbf{U}_c[j]$, the sample error for that bin and each of its offsets $\mathbf{W}_c[j]$ have the same variance $T_c[j]\cdot\nu^2$.
\end{lemma}

\begin{proof}
Since this statement is based on the premise that all of the earlier bin detector runs were accurate, we can assume that the index $\widehat{m}=m$ is correct. 
We will still write $\widehat{m}$ to distinguish these index estimates from the original index $m$. 
The idea is to calculate the variance after a peel by induction with the fact that for any two random variables,
\begin{align*}
    \Var{X+Y}=\Var{X}+\Var{Y}+2\Cov{X,Y}.
\end{align*}
Assume that the statement in the lemma holds before a peeling. 
Then we need to show that if we subtract an estimated Pauli $\widehat{p}_{\widehat{m}}$ from a bin in a different subsampling group that contains this Pauli, the statement is preserved when updating $\mathbf{T}_c[j]$. 
To do this we will work out the variance of the sample error in each term and the covariance between these sample errors. 
Armed with this we will be able to prove that the statement is preserved after peeling.

The peeling process (Line~\ref{algline:U} in \Cref{alg:peeling}) causes error propagation for the error part of the bin $\mathbf{U}_c[j]$ as follows,
\begin{gather}\label{update}
    W_{c;t}[j]\leftarrow W_{c;t}[j]+(p_m-\widehat{p}_{\widehat{m}})(-1)^{\langle\mathbf{d}_{c;t},\widehat{m}\rangle}.
\end{gather}
The variance of the first term is by induction $\mathbf{T}_c[j]\nu^2$, while the variance of the second term is not so trivial. 
The estimated Pauli comes from the \textit{single-ton search}, where all the first $P_1$ observations get summed after adding a random sign. 
Since all the random signs of the $W_{\widehat{m}}$ terms are annihilated before summation, and all the other error parts still remain random, the variance of this second term in \eqref{update} is
\begin{align}\label{eq:variance of hat_p}
    \Var{p_m-\widehat{p}_{\widehat{m}}}=\frac{T_{c'}[j']\times\nu^2}{P_1}+\frac{(P_1-1)B\times\nu^2}{P_1N}.
\end{align}
Because of the assumption that all the initial errors in different bins are independent and the condition that the every peeling route is cycle-free, the covariance term vanishes. 
Therefore, we have calculated variance as follows
\begin{align*}
    \Var{W_{c;t}[j]}_{\text{after}}/\nu^2=T_{c}[j]+\frac{\mathbf{T}_{c'}[j']}{P_1}+\frac{(P_1-1)B}{P_1N},
\end{align*}
which proves the lemma.
\end{proof}

In order to prove \Cref{Thm:main} and find a bound on the variances of the sample errors, it is necessary to find an upper bound on the parameters $\mathbf{T}_c[j]$, which need to be analyzed for both the graph and the algorithm. 
Denote by $G$ the bipartite graph of which each right node represents a bin observation, each left node represents a nonzero Pauli error rate, and edges come from the hash function relation:
\begin{gather}
    \mathbf{M}^T_cm=j.
\end{gather}
That is, a bin-observation-node $\mathbf{U}_c[j]$ is connected to error-rate-node $p_m$ if and only if it holds that $\mathbf{M}^T_cm=j$. 
A right node is a single-ton node if and only if it has a single edge connected to it. 
Every time we peel a left node (that is, we identify a Pauli error) we remove the edges connecting it and the right nodes. 
Each peeling therefore decreases the degree of the right nodes.

The following two lemmas help us bound the integer array $\mathbf{T}$ as we peel along the graph $G$. 
We first bound the right degree of $G$.

\begin{lemma}\label{lem:Bounded_degree}
The maximum degree of the right nodes in $G$ is less or equal than $\frac{P_1}{2}$ with probability $1-\mathrm{e}^{-O(n)}$.
\end{lemma}

\begin{proof}
Put the right nodes in some sequential order, and define events $\{X_i\}_{i=1}^{BC}$ where $X_i$ denotes the $i^{\text{th}}$ node and is linked to more than $P_1/2$ left nodes. 
According to the bin observation model \eqref{eq:bin_description}, each bin connects with $\frac{N}{B}$ Pauli error rates (most of which will be zero), so the expected degree of a right node is $\frac{s}{B}$ where $s$ is the number of left nodes. 
Since the algorithm chooses $B=O(s)$, the expected degree $\frac{s}{B}= O(1)$.

Note that by Assumption~\ref{random_support_assumption}, the support of the Pauli error rates is chosen randomly. 
Therefore, concentrating on a specific bin $i$, we can introduce a random variable $d_i$ that denotes the degree of bin-observation-node $i$ (a right node), and introduce the variable $d_{ij}$ which is 0 if the corresponding $j^{\text{th}}$ Pauli error rate is zero or if $p_j$ is not in the $i^{\text{th}}$ bin, and 1 otherwise. 
Then we have the relation
\begin{align}
    d_i = \sum_{j} d_{ij},
\end{align}
since this counts the support in the $i^{\text{th}}$ bin. 
These variables $d_{ij}$ are actually all Bernoulli variables and the only correlation among them comes from the constraint that there are exactly $s$ elements in the entire support.
This constraint means that the $d_{ij}$ are negatively correlated, and so the probability of $d_{ij}=1$ can be upper bounded by considering the event that all the other Pauli rates linked with this bin are zero,
\begin{gather*}
    \Prob{d_{ij}=1}\leq\frac{sB}{N(B-1)}.
\end{gather*}
Now consider another set of i.i.d.\ Bernoulli variables $\{d_{ij}'\}_j$ each of which is 1 with probability $\frac{sB}{N(B-1)}$.
We then have
\begin{gather*}
    \Prob{X_i}=\Prob{\sum_{j=1}^{N/B}d_{ij}\geq \frac{P_1}{2}}\leq\Prob{\sum_{j=1}^{N/B}d_{ij}'\geq\frac{P_1}{2}}.
\end{gather*}
Since the expected value of the sum of $\{d_{ij}'\}$ is $\frac{s}{B-1}$, the Chernoff bound is suitable for this case, and we find
\begin{align*}
    \Prob{X_i}\leq&\Prob{\sum_{j=1}^{N/B}d_{ij}'-\frac{s}{B-1}\geq \frac{P_1}{2}-\frac{s}{B-1}}\\
    \leq&\mathrm{e}^{-\frac{[(B-1)P_1-2s]^2}{2(B-1)[(B-1)P_1+2s]}}.
\end{align*}
According to the union bound, the event $X$ which denotes that there exist some left nodes with degree larger than $\frac{P_1}{2}$ follows from the upper bound,
\begin{align*}
    \Prob{X}\leq BC\Prob{X_i} \leq BC\mathrm{e}^{-\frac{[(B-1)P_1-2s]^2}{2(B-1)[(B-1)P_1+2s]}}.
\end{align*}
That this is at most $\mathrm{e}^{-O(n)}$ follows since $B = \Theta(s)$, $s = \Theta(N^\delta)$, and $P_1 = \Theta(n)$. 
\end{proof}

The degree bound we have just proven allows us to bound the maximum element of the array $\mathbf{T}$. 
\begin{lemma}\label{maximum_T}
Suppose the maximum degree of the right nodes in $G$ is not larger than $\frac{P_1}{2}$. 
Then for any time step in \Cref{alg:peeling}, assuming all previous bin detections succeeded, the maximum element of the array $\mathbf{T}$ is at most 4.
\end{lemma}
\begin{proof}
The recursive equation for $\mathbf{T}$ (Line~\ref{algline:T} in \Cref{alg:peeling}) is
\begin{gather*}
    \mathbf{T}_{c'}[j_{c'}]\leftarrow\mathbf{T}_{c'}[j_{c'}]+\frac{\mathbf{T}_c[j]}{P_1}+\frac{(P_1-1)B}{P_1N}.
\end{gather*}
The algorithm defines a time sequential order for each nonzero Pauli error rate to be detected.
For each step $i$, let $\mathbf{U}_{z_i}$ be the next bin in which we will find a nonzero Pauli rate. 
Let $T_{\max}$ be the current maximum element in a subset $\mathbf{T}_{\text{peeled}}$ of $\mathbf{T}$. 
Then $\mathbf{T}_{\text{peeled}}$ contains those $\mathbf{T}_c[j]$ of bins in which we have already found a nonzero Pauli error.
Also, we use $T_{\max}'$ to indicate the maximum $\mathbf{T}_c[j]$ that will exist after the next step.

According to the assumption, the maximum degree of an arbitrary right node is less than $\min(P_1/2,N/B)$, and the number of peels needed is never more than the maximum degree of that node. 
Note we assume that the previous bins have been accurately detected, so the process will always choose nonzero Pauli error rates (and the corresponding bins) to peel. 
The noise in the peeling bin will increase by at most $\frac{T_{\max}}{P_1}+\frac{(P_1-1)B}{P_1N}$. 
Each $\mathbf{T}$ is initialized as 1, and denote the number of peelings by $\kappa$. 
Therefore
\begin{align*}
    T_{\max}'&\leq 1+\kappa\cdot\left(\frac{T_{\max}}{P_1}+\frac{(P_1-1)B}{P_1N}\right)\\
    &\leq 1+\frac{T_{\max}}{2}+1\,.
\end{align*}
Then by induction, because initially the maximum element in $\mathbf{T}_{\text{peeled}}$ is equal to 0 and in each step this value will increase via the above formula, the maximum element we can get is the limit of this recursive inequality, which is  $T_{\max}\leq4$.
\end{proof}

According to the above two lemmas, the integer $T_{\max}$ indicates that the upper bound of the maximum element in the array $\mathbf{T}$ is not larger than 4 with high probability for a sufficiently large $N$.
In order to compute the failure rates and make the algorithm realizable, we assumed (as in Assumption \textbf{A3}) that the minimum nonzero Pauli error rate $\epsilon_0$ satisfies $\epsilon_0^2\geq4\nu^2=4\xi^2/B$, which makes a distinctive barrier between Pauli error rates and any noise.

In anticipation of applying the union bound, let us define the following error categories and their probabilities.
For brevity, we will denote a zero-ton, single-ton, or multi-ton by just the letters \texttt{z}, \texttt{s}, or \texttt{m}, and we denote the true value of the bin by $\mathfrak{B}$.

\begin{definition}[Failure modes for bin detection]\label{def_error_category}
The bin detection algorithm failure modes are defined as follows:
\begin{itemize}
    \item The single-ton false negative probability:
    \begin{align*}
        \Pr(\mathrm{SFN}) := \Pr(\widehat{\mathfrak{B}} = \texttt{z}\, |\, \mathfrak{B} = \texttt{s})+\Pr(\widehat{\mathfrak{B}} = \texttt{m}\, |\, \mathfrak{B} = \texttt{s})
    \end{align*}
    \item The  single-ton false positive probability:
    \begin{align*}
        \Pr(\mathrm{SFP}) := \Pr(\widehat{\mathfrak{B}} = \texttt{s}\, |\, \mathfrak{B} = \texttt{z})+\Pr(\widehat{\mathfrak{B}} = \texttt{s}\, |\, \mathfrak{B} = \texttt{m})
    \end{align*}
    \item The multi-ton $\leftrightarrow$ zero-ton confusion probability:
    \begin{align*}
        \Pr(\mathrm{MZ}) := \Pr(\widehat{\mathfrak{B}} = \texttt{z}\, |\, \mathfrak{B} = \texttt{m})+\Pr(\widehat{\mathfrak{B}} = \texttt{m}\, |\, \mathfrak{B} = \texttt{z})
    \end{align*}
    \item The index error probability:
    \begin{align*}
        \Pr(\mathrm{I}) :=\Pr(\widehat{\mathfrak{B}}=\texttt{s},\widehat{m}\not=m\, |\,\mathfrak{B}=\texttt{s}, m)
    \end{align*}
    \item The value error probability:
    \begin{align*}
        \Pr(\mathrm{V}) :=\Pr(\widehat{\mathfrak{B}}=\texttt{s},|\hat{p}_{\hat{m}} - p_m|> 2\xi/\sqrt{B} \,|\,\mathfrak{B}=\texttt{s}, m, p_m)
    \end{align*}
\end{itemize}
\end{definition}

Of course these probabilities are not all independent. 
However, by the union bound it suffices to bound each of these bad events individually and the total failure probability will be at most the sum of the probabilities of these failure modes. 
We will show that all of these failure probabilities decay exponentially with $P_1$, the number of randomly chosen offsets. 

\begin{lemma}\label{lm:failure_bound}
Let $E$ denote the event that an arbitrary bin detection with inputs as those in \Cref{alg:peeling} returns either the wrong bin type, the wrong index or an estimated Pauli error rate with error larger than $2\nu = 2\xi/\sqrt{B}$. 
Let $D$ be the event that the maximum degree of the right nodes in $G$ is not larger than $P_1/2$. 
Let $V$ be the event that all prior bin detections succeeded. 
And denote the event that every peeling route is cycle-free by $H$.
Then
\begin{gather}
    \Prob{E|D,V,H} \leq O(1) \mathrm{e}^{-O(n)}.
\end{gather}
\end{lemma}

\begin{proof}
This theorem means that the bin detector algorithm succeeds with high probability whenever $D$, $V$ and $H$ occur. 
To show it, we have to bound the failure probabilities for each failure mode of the bin detection algorithm and then apply the union bound. 
We will prove most of our statements by bounding failure probabilities with expressions of the form $\mathrm{e}^{-O(P_1)}$. 
This is equivalent to a bound of the form $\mathrm{e}^{-O(n)}$ since $P_1 = \Theta(n)$ by \Cref{SO}. 
Also note that conditioning on events $D$, $V$ and $H$ allows the use of \Cref{lm:variance} and \Cref{maximum_T}, specifically that the variance of noise in each row of this bin is $T\nu^2$ from \Cref{lm:variance}, and that this $T$ is no more than $4$ according to \Cref{maximum_T}. 

We first consider the \textit{single-ton false negative} probability in \Cref{def_error_category}. 
Note in this case the underlying bin contains only one Pauli error rate along with noise, that is
\begin{gather}\label{underlying}
    \mathbf{U}_c = p_m (-1)^{\langle \mathbf{D}_c,m\rangle} +\mathbf{W}.
\end{gather}
Let $f_1 = \Prob{\widehat{\mathfrak{B}}=\texttt{z}\,|\,\mathfrak{B}=\texttt{s}}$. 
Then by Line~\ref{algline:zero-ton} in \Cref{alg:bin_detect}, the probability can be upper bounded by the probability of a single-ton bin passing the zero-ton verification:
\begin{align*}
	&f_1
	\leq \Prob{\frac{1}{P_1}\left\|p_m \mathbf{s}_{c,m}+\mathbf{W}\right\|^2\leq T(1+\gamma_1)\nu^2},
\end{align*}
where $\mathbf{s}_{c,m}$ is the vector such that $\mathbf{s}_{c,m}[t] = (-1)^{\langle m,d_{c;t}\rangle}$. 
Since here the noise vector $\mathbf{W}$ comes from the sum of noise $w$, it's obvious that all the elements of $\mathbf{W}$ are Gaussian distributed with variance $T\nu^2$.
Therefore, according to the tail bounds of \Cref{lem:tailbound}, the following holds as long as $\gamma_1<\epsilon_0^2/T\nu^2$.
\begin{align}
	f_1
	\leq \mathrm{e}^{-\frac{P_1}{4}\frac{\left(\epsilon_0^2/T\nu^2-\gamma_1\right)^2}{1+2\epsilon_0^2/T\nu^2}}.
\end{align}

Now let $f_2 = \Prob{\widehat{\mathfrak{B}}=\texttt{m}|\mathfrak{B}=\texttt{s}}$. 
This kind of failure happens if and only if the single-ton bin fails during single-ton verification,
\begin{gather*}
    	f_2=\Prob{\frac{1}{P_1}\left\|\mathbf{U}_c-\widehat{p}_{\widehat{m}}\mathbf{s}_{c,\widehat{m}}\right\|^2 \geq  T(1+\gamma_2)\nu^2}.
\end{gather*}
Considering the underlying structure of this bin \eqref{underlying}, this probability can be bounded using a conditional probability. 
We first denote the event $\{|\widehat{p}_{\widehat{m}} - p_m|>\sqrt{4\nu^2}~\textrm{or}~\widehat{m}\neq m\}$ by $E_0$. 
Then we observe 
\begin{align*}
	f_2&\leq
	\Prob{E_0}	+\Prob{\frac{1}{P_1}\left\|\mathbf{U}_c-\widehat{p}_{\widehat{m}}\mathbf{s}_{c,\widehat{m}}\right\|^2 \geq  T(1+\gamma_2)\nu^2 \Big|E_0^c}.
\end{align*}
Using the tail bound \eqref{tailbound1}, we have that
\begin{align}
	&\Prob{\frac{1}{P_1} \left\|\mathbf{U}_c-\widehat{p}_{\widehat{m}}\mathbf{s}_{c,\widehat{m}}\right\|^2 \geq  T(1+\gamma_2)\nu^2 \Big|E_0^c}\notag\\
	&\leq \mathrm{e}^{-\frac{P_1}{4}\left(\sqrt{1+2\gamma_2}-\sqrt{1+2\times4/T}\right)^2}.\label{etabound}	
\end{align}
Then using union bound, we can deal with the first term
\begin{align}
	\Prob{E_0}&\leq
	\Prob{|\widehat{p}_{\widehat{m}} - p_m|>\sqrt{4\nu^2}}+\Prob{\widehat{m}\neq m}\notag\\
	&\leq\Prob{|\widehat{p}_{\widehat{m}} - p_m|>\sqrt{4\nu^2}\Big|\widehat{m}=m}
	+2\Prob{\widehat{m}\neq m}.\label{prob law}
\end{align}
Note above, the estimated Pauli error rate can be calculated according to \Cref{alg:bin_detect}, so we obtain the bound
\begin{align}
	&\Prob{|\widehat{p}_{\widehat{m}} - p_m|>\sqrt{4\nu^2}\Big|\widehat{m}=m}\notag\\
	=&\Prob{\left|\frac{\mathbf{s}_{c,m}^T\mathbf{U}_c}{P_1}-p_m\right|>\sqrt{4\nu^2}}\notag\\
	=&\Prob{|Y|/P_1>\sqrt{4\nu^2}}\leq2\mathrm{e}^{-\frac{4P_1}{2T}},\label{bound}
\end{align}
where $Y$ is the sum of $P_1$ i.i.d.\ Gaussian variables with $\mathcal{N}\left(0,\left(T+\frac{(P_1-1)B}{N}\right)\nu^2\right)$ like \eqref{eq:variance of hat_p}, and the last inequality comes from the Chernoff-Hoeffding bound \cite{Hoeffding1963}. 
According to \Cref{maximum_T}, exponents in \eqref{etabound} and \eqref{bound} are both scaling linearly with $P_1$, thus the probabilities decay exponentially with $P_1$.

Since the second term in \eqref{prob law}, $\Prob{\widehat{m}\neq m}$, is essentially the probability of the index error, the failure probability of such a decoding process also decays exponentially with $P_2$. 
In accordance with \eqref{eq:sign_code}, the sign vector $[\sgn{U_{P_1}},\cdots,\sgn{U_{P-1}}]^T$ is the sum of a codeword $\langle\mathbf{G},m\rangle$ and a vector of noise. 
Since the decoding process fails only if the weight of the noise is larger than the code distance $\beta P_2$ and each element of the noise is an independent Bernoulli random variable with error probability upper bounded by $\mathbb{P}$, the index error probability can be bounded by the Chernoff-Hoeffding bound:
\begin{gather}\label{crossed_error}
    \Prob{\widehat{m}=m}\leq \mathrm{e}^{-\frac{(\beta/\mathbb{P}-1)^2}{3}P_2}.
\end{gather}
Moreover, as noted in Remark~\ref{strictly_small}, we have $\mathbb{P}_m< \frac{1}{2}$ for all $m\in\mathbb{F}_2^{2n}$. 
Given the assumptions of $D$, $V$ and $H$, we choose the maximum $\mathbb{P}_m$ to be $\mathbb{P} = \max_{m} \mathbb{P}_m$.
Therefore, using the law of total probability we have
\begin{align}
    f_2\leq& \mathrm{e}^{-\frac{P_1}{4}\left(\sqrt{1+2\gamma_2}-\sqrt{1+2\times4/T}\right)^2}\notag
    +2\mathrm{e}^{-\frac{4P_1}{2T}}\\
    &+2\mathrm{e}^{-\frac{(\beta/\mathbb{P}-1)^2}{3}P_2}.\label{eq:f_2}
\end{align}
Recall from \Cref{SO} that $P_1$ and $P_2$ are proportional to $n$, so we have a bound $f_2 = \mathrm{e}^{-O(n)}$. 

We now turn to the case that the bin detection algorithm incorrectly recognizes a zero-ton or a multi-ton bin as a single-ton bin, i.e., we consider the \textit{single-ton false positive} probability. 
For this, we need to consider the general underlying bin structure
\begin{align}
	\mathbf{U}_c = \mathbf{S}_c\mathbf{p}+\mathbf{W},
\end{align}
where $\mathbf{U}_c$ is either zero-ton or multi-ton, and only contains the $P_1$ fully random offsets when choosing as in \Cref{SO}. 
Here $\mathbf{S}_c\in\{\pm 1\}^{P_1\times N/B}$ is the sign matrix constructed according to \Cref{lm:rewritten_prop_hashing_obs}.

Now consider the probability of the bin detector falsely detecting a zero-ton as a single-ton, and denote $\Prob{\widehat{\mathfrak{B}}=\texttt{s}|\mathfrak{B}=\texttt{z}}$ by $f_3$. 
By Line~\ref{algline:zero-ton} in \Cref{alg:bin_detect}, the probability of $f_3$ can be bounded by the probability of zero-ton verification failing
\begin{gather*}
    f_3\leq \Prob{\frac{1}{P_1}\left\|\mathbf{W}\right\|^2\geq T\times (1+\gamma_1)\nu^2}.
\end{gather*}
According to the tail bound \eqref{tailbound1}, this failure probability can be bounded by an exponentially decaying function
\begin{align}
    f_3\leq \mathrm{e}^{-\frac{P_1}{4}(\sqrt{1+2\gamma_1}-1)^2}.
\end{align}

Now let $f_4 = \Prob{\widehat{\mathfrak{B}}=\texttt{s}|\mathfrak{B}=\texttt{m}}$. 
This error probability can be evaluated under the multi-ton model when it passes the single-ton verification step for some estimated index-value pair $(\widehat{m},\widehat{p}_{\widehat{m}})$. 
Using Line~\ref{algline:single-ton_verification} in \Cref{alg:bin_detect}, 
\begin{gather*}
    f_4\leq\Prob{\frac{1}{P_1}\left\|\mathbf{U}_c-\widehat{p}_{\widehat{m}}\mathbf{s}_{c,\widehat{m}}\right\|^2 \leq T\times(1+\gamma_2)\nu^2}.
\end{gather*}
Let 
\begin{align}
    \mathbf{g} = \mathbf{S}_c\mathbf{p}-\widehat{p}_{\widehat{m}}\mathbf{s}_{c,\widehat{m}},
\end{align}
and let the sample error be $\mathbf{W} = \mathbf{k}$.
Then the law of total probability can be used as follows:
\begin{align}
    f_4=&\Prob{\frac{1}{P_1}\left\|\mathbf{g}+\mathbf{k}\right\|^2\leq T(1+\gamma_2)\nu^2}\notag\\
    \leq& \Prob{\frac{1}{P_1}\left\|\mathbf{g}+\mathbf{k}\right\|^2 \leq  T(1+\gamma_2)\nu^2\Big|\frac{\left\|\mathbf{g}\right\|^2}{P_1}\geq 2T\gamma_2\nu^2}\notag\\
	&\ +\Prob{\frac{\left\|\mathbf{g}\right\|^2}{P_1}\leq 2T\gamma_2\nu^2}.
\end{align}
Note that the first term can be bounded by \eqref{tailbound2} since the conditional part shows the lower bound of the parameter $\theta_0$ as defined in \Cref{lem:tailbound}
\begin{align}
    &\Prob{\frac{1}{P_1}\left\|\mathbf{g}+\mathbf{k}\right\|^2 \leq  T(1+\gamma_2)\nu^2\Bigg|\frac{\left\|\mathbf{g}\right\|^2}{P_1}\geq 2T\gamma_2\nu^2}\notag\\
    &\leq \mathrm{e}^{-\frac{P_1}{4}\frac{\gamma_2^2}{1+4\gamma_2}}.
\end{align}
The second term can be bounded as follows. 
Let $\alpha = \mathbf{p}-\widehat{p}_{\widehat{m}}\mathbf{e}_{\widehat{m}}$, and we have
\begin{gather*}
    \Prob{\frac{\left\|\mathbf{g}\right\|^2}{P_1}\leq 2T\gamma_2\nu^2}
	=\Prob{\frac{\left\|\mathbf{S}_c\mathbf{\alpha}\right\|^2}{P_1}\leq 2T\gamma_2\nu^2}.
\end{gather*}
Here $\mathbf{e}_{k}$ is the vector with support only on the $k^\text{th}$ element.
We denote the support set of the vector $\alpha$ by $\mathcal{L}_0$, and define the $\epsilon_0$-essential support of $\mathbf{\alpha}$ to be 
\begin{align}
    \mathcal{L} = \bigl\{ i \in \mathbb{F}_{2}^{2n}\,\big|\, |\alpha_i| \geq \epsilon_0 \bigr\}.
\end{align}
Denote the cardinality of $\mathcal{L}$ by $L$. 
Then the above probability can be bounded by an application of the Chernoff-Hoeffding bound.

With the same argument as in the proof of \Cref{lm:independence}, the sample error in each row in vector $\mathbf{g}$ is independent, and so is the square of that error. 
When we calculate $\|\mathbf{g}\|^2$, we can regard it as a sum of $P_1$ independent random variables. 
Also, each term in this sum contains the same structure, and identically distributed parameters, so we can claim each term is identically distributed. 

Therefore, we first analyze the expected value $E$ of each variable in this sum. 
Take one of these terms $X_i$ as an example,
\begin{gather}\label{eq:definition_of_term}
    X_i:=\left(\sum_{j\in\mathcal{L}_0}(-1)^{\langle d_i,j\rangle}\alpha_j\right)^2,
\end{gather}
where $\{d_i\}$ is a set of independent random $2n$-bit strings.
The expected value $E$ of $X_i$ satisfies the following bound
\begin{align}
    E=\mathbb{E}(X_i)\geq L\epsilon_0^2.
\end{align}
Note that above we used the fact that any random strings are independent.

Moreover, since we want to show this term will be large with high probability, we should consider the random cross terms in each $X_i$, and that is
\begin{align}
    R_i=\sum_{\substack{u>v\\u,v\in\mathcal{L}_0}}(-1)^{\langle d_i,u+v\rangle}2\alpha_u\alpha_v,
\end{align}
where the order is in lexicographical order. 
Note the remaining part of $X_i$ is a deterministic one, so we only calculate the variance for this $R_i$. 
It's straightforward that we have
\begin{align}
    \Var{R_i}=\mathbb{E}[R_i^2],
\end{align}
and the only contributed terms are those without random signs in $R_i^2$.
For example, if we consider a specific $u,v$ and the term $(\alpha_u\alpha_v)\times(\alpha_w\alpha_x)$ with some $(w,x)\neq(u,v)$ (assume $w>x$), this term will contribute to the expected value only if $w+x+u+v=0$. 
Therefore, the only potential effective terms are those with four different Pauli error rates.
Moreover, since $w+x+u+v$ must be in the null-space of $M_c^T$ according to \Cref{lm:prop_hashing_obs}, of which the size is $\frac{N}{B}=\frac{1}{\eta}\mathrm{e}^{(1-\delta)n}$, we can estimate this probability using a Balls and Bins model.

Regardless of the square terms, the number of potential terms is $\binom{L}{4}$. 
Let $G_i$ be the probability that the number of terms in bin 0 is at least a chosen constant $\eta_0$. 
Then $G_i$ can be bounded as
\begin{align*}
    \Prob{G_i}\leq \binom{L}{4}\times\left(\frac{B}{N}\right)^{\eta_0}=\binom{L}{4}\cdot\frac{\mathrm{e}^{-(1-\delta)n\eta_0}}{\eta},
\end{align*}
which is decaying exponentially with $n$. 
Given that the complementary event $G_i^c$ happens and defining the set of contributing terms to be $\mathcal{A}_i$, the variance of $R_i$ is the sum
\begin{align*}
    \Var{R_i}=&\sum_{\substack{(u,v,w,x)\in\mathcal{A}_i\\u>v>w>x}}8\alpha_u\alpha_v\alpha_w\alpha_x+\sum_{\substack{u>v\\u,v\in\mathcal{L}_0}}4\alpha_u^2\alpha_v^2
    \end{align*}
Then by Cauchy-Schwarz we have $(\alpha_u \alpha_v)^2 + (\alpha_w \alpha_x)^2 \geq 2 \alpha_u \alpha_v \alpha_w \alpha_x$. 
Averaging this over the other distinct partitions and using $|\mathcal{A}_i| \le\eta_0$, we find
\begin{align*}
    \Var{R_i}\leq& \frac{\eta_0+3}{3}\sum_{\substack{u>v\\u,v\in\mathcal{L}_0}}4\alpha_u^2\alpha_v^2.
\end{align*}
Now we can use the Hoeffding bound to obtain
\begin{align*}
    &\Prob{\frac{\left\|\mathbf{g}\right\|^2}{P_1}\leq 2T\gamma_2\nu^2}\\ 
   \leq &\Prob{\frac{\left\|\mathbf{g}\right\|^2}{P_1}\leq 2T\gamma_2\nu^2\Bigg|G^c}+\Prob{G}\\
    \leq&\mathrm{e}^{-\frac{3P_1(L\epsilon_0^2-2T\gamma_2\nu^2)^2}{2(3+\eta_0)L^2\epsilon_0^4}}+O\left(P_1L^4\mathrm{e}^{-(1-\delta)n\eta_0}\right).
\end{align*}
The last inequality uses the fact that 
\begin{align*}
    \mathbb{E}[X_i]^2\geq \sum_{\substack{u>v\\u,v\in\mathcal{L}_0}}4\alpha_u^2\alpha_v^2\geq2L(L-1)\epsilon_0^4\,.
\end{align*}
For any nontrivial signal, we have that $1\leq L<P_1/2$. 
As long as $0<\gamma_2<\epsilon_0^2/2T\nu^2$ and choosing $\eta_0=6$, for any multi-ton we have 
\begin{align}
 	f_4\leq&\mathrm{e}^{-\frac{P_1}{4}\frac{\gamma_2^2}{1+4\gamma_2}}+\mathrm{e}^{-\frac{P_1(\epsilon_0^2-2T\gamma_2\nu^2)^2}{6\epsilon_0^4}}+O\left(\mathrm{e}^{-6(1-\delta)n}\right).
\end{align}
Therefore $f_4 \le O(1) \mathrm{e}^{-O(n)}$. 

Next we will consider the \textit{multi-ton--zero-ton confusion} probability
\begin{align*}
    \Pr(\mathrm{MZ}) := \Pr(\widehat{\mathfrak{B}} = \texttt{z}\, |\, \mathfrak{B} = \texttt{m})+\Pr(\widehat{\mathfrak{B}} = \texttt{m}\, |\, \mathfrak{B} = \texttt{z}).
    \end{align*}
Denote the first term $\Pr(\widehat{\mathfrak{B}} = \texttt{z}\, |\, \mathfrak{B} = \texttt{m})$ by $f_5$ and the second $\Pr(\widehat{\mathfrak{B}} = \texttt{m}\, |\, \mathfrak{B} = \texttt{z})$ by $f_6$. 
For $f_5$, recognizing a multi-ton as a zero-ton, we have the following inequality, 
\begin{align*}
    f_5\leq\Prob{\frac{1}{P_1}\left\|\mathbf{U}\right\|^2 \leq T\times(1+\gamma_1)\nu^2}.
\end{align*}
Note this probability can be analyzed in just the same way as $f_4$, and the only difference is that when we consider $f_5$, the $\alpha$ is just based on several underlying Pauli error rates without any subtraction, so $L\geq 2$ for this case. 
As long as $0<\gamma_1<\epsilon_0^2/T\nu^2$, then for any multi-ton we have the bound
\begin{align}
    f_5\leq&\mathrm{e}^{-\frac{P_1}{4}\frac{\gamma_1^2}{1+4\gamma_1}}+\mathrm{e}^{-P_1\frac{(\epsilon_0^2-T\gamma_1\nu^2)^2}{6\epsilon_0^4}}+O\left(\mathrm{e}^{-6(1-\delta)n}\right),
\end{align}
so $f_5 \leq O(1) \mathrm{e}^{-O(n)}$. 
Moreover, it is clear that the failure probability of recognizing a zero-ton bin as a multi-ton bin, namely $f_6$, is smaller than $f_3$.

Next, consider the \textit{index error} probability, and denote $\Prob{\widehat{\mathfrak{B}}=\texttt{s},\widehat{m}\neq m|\mathfrak{B}=\texttt{s},m}$ by $f_7$. 
This probability can be bounded by the probability of estimating a wrong index $\widehat{m}$ and some Pauli error rate, and still passing the single-ton verification 
\begin{align}
    f_7&\leq\Prob{(\widehat{m}\neq m)~\wedge~(\widehat{m},\widehat{p}_{\widehat{m}}){\rm\ passes\ verification}}\notag\\
    &\leq \Prob{\widehat{m}\neq m}\leq  \mathrm{e}^{-\frac{(\beta/\mathbb{P}-1)^2}{3}P_2}. \label{eq:f_7}
\end{align}
Note the last inequality is just \eqref{crossed_error}, and according to Remark~\ref{strictly_small}, $\mathbb{P}_m< \frac{1}{2}$ for all $m\in\mathbb{F}_2^{2n}$ given that events $D$ and $V$ happen and we choose the maximum one to be $\mathbb{P}$.

Finally, let's consider the \textit{value error} probability, and denote $\Pr(\widehat{\mathfrak{B}}=\texttt{s},|\hat{p}_{\hat{m}} - p_m|> \sqrt{4\nu^2} |\mathfrak{B}=\texttt{s}, m, p_m)$ by $f_8$. 
Note that we have chosen $\sqrt{4\nu^2}$ as the error bound for the Pauli error rate, so similar to the \textit{index error} probability, this $f_8$ can be bounded by the probability of estimating a noisy Pauli error rate and passing the single-ton verification. 
We can loosen this bound by only considering the first event, and we obtain the inequality
\begin{align}
    f_8&\leq\Prob{|\widehat{p}_{\widehat{m}} - p_m|>\sqrt{4\nu^2}\Big|\widehat{m}=m}
	+\Prob{\widehat{m}\neq m}\notag\\
	&\leq2\mathrm{e}^{-\frac{4P_1}{2T}}+\mathrm{e}^{-\frac{(\beta/\mathbb{P}-1)^2}{3}P_2} \leq 3 \mathrm{e}^{-O(n)}. \label{eq:f_8}
\end{align}
Note the middle inequality comes from a combination of \eqref{bound} and \eqref{crossed_error}. 
According to Remark~\ref{strictly_small}, we again have $\mathbb{P}_m< \frac{1}{2}$ for all $m\in\mathbb{F}_2^{2n}$ given that events $D$, $V$ and $H$ happen and we choose the maximum one to be $\mathbb{P}$.

Following the taxonomy in \Cref{def_error_category}, we have treated all of the failure cases of the bin detector algorithm. 
Using the union bound, we can get the following inequality
\begin{gather*}
    \Prob{E}\leq\sum_{i=1}^8f_i.
\end{gather*}
As we illustrated at the beginning of this proof, events $D$, $V$ and $H$ have shown that the variance of the noise in each row of this bin is $T\nu^2$ from \Cref{lm:variance}, and that this $T$ is no more than $4$ according to \Cref{maximum_T}. 
Furthermore, they imply that $\theta_m$ is strictly smaller than $\frac{1}{2}$ for all $m\in\mathbb{F}_{2}^{2n}$ in \eqref{eq:f_2}, \eqref{eq:f_7} and \eqref{eq:f_8}. 
Since constraining the peeling graph $G$ to obey this event is independent of the above analysis of the failure probabilities of the bin detector, it follows that
\begin{gather*}
    \Prob{E|D,V,H}\leq O(1) \mathrm{e}^{-O(P_1)}.
\end{gather*}
This completes the proof.
\end{proof}

\section{Conclusion}\label{sec:conclusion}

We have shown that for sparse Pauli channels we can learn all the significant Pauli errors, even those associated with high-weight Pauli strings, using realistic experimental resources that scale with the sparsity of the Pauli errors rather than the dimension. 
In particular we have demonstrated that using only a few local two-qubit gates and a number of quantum experiments that scales linearly (with a factor of about 4), we can recover up to $4^{\delta n}$ of the largest error rates, where $\delta \lesssim 0.25$. 
Our numerical analysis indicates in the regime where  $0.25 < \delta < 0.5$ we can still recover these errors with a number of experiments that only scales as $O(n^2)$.

We support these experimental protocols by defining and analyzing an algorithm with rigorous performance guarantees. 
This provable algorithm confirms that, with explicitly stated assumptions, high-precision reconstruction is possible when querying only a number of Pauli eigenvalues that scales like $O(sn)$. 
Moreover, the heuristic practical circuits used above are able to approximate the relevant noisy eigenvalue queries with sufficient precision $\xi$ using only $O(n^2/\xi^2)$ measurements. 
These circuits exploit the protocols presented in \cite{Flammia2019} and \cite{Harper2019} to learn up to $2^n$ commuting Pauli eigenvalues per experiment, and greatly reduces the required experimental resources.

This work provides an experimentally realizable method of identifying the relevant Pauli errors in large-scale quantum devices even if there are unexpected long-range correlations between the qubits. 
The ability to do so will be vital as we seek to mitigate the errors in such devices, to learn the noise patterns that exist when such devices are operated holistically and will allow better designing and tailoring of error correction and fault tolerance in such devices.

Many interesting open questions remain. 

For example, what about very large scale devices? 
The practicality of the algorithm in the regime of greater than (say) $30$ qubits, where memory storage becomes an issue, could potentially be addressed as follows. 
We keep the protocol executed on the device identical as system size increases. 
However, we can take advantage of the fact that the WHT commutes with the marginalization of the observed probabilities and the process of fitting required to ascertain the SPAM-free eigenvalues (see \cite{Harper2019}). 
The actual observations only require $n$ bits of data to store. 
We can, therefore, marginalize the observations to obtain overlapping sets of $2^m$ eigenvalues (where we choose $m$ to be the largest computationally tractable number for our classical computer). 
This will mean that we have multiple sets of $2^m$ bins, each potentially containing $2^{2n-m}$ Pauli error rates.
Given this, the $s$-sparse assumption now becomes $s<2^m$. 
It would be extremely interesting to implement this version of the algorithm on real data. 

Another approach to dealing with very large scale devices is to incorporate our algorithm as a subroutine in a larger algorithm that builds a globally consistent Pauli error distribution from estimations of marginal error rates. 
For example, as proposed in Ref.~\cite{Flammia2019}, one could efficiently estimate a Markov random field description of a Pauli channel if the underlying graphical model has bounded degree correlations. 
This idea has been performed experimentally on $14$ qubits in Ref.~\cite{Harper2019}. 
We believe using the algorithm presented here would improve the estimation of the core subroutines and lead to better performance of the global reconstruction. 

There are also several open mathematical questions about the reconstruction of sparse (or approximately sparse) Pauli channels. 
For example, it would be interesting to relax the random sparsity assumption on the support, or to allow for prior information in the distribution of the support. 
It would also be interesting to treat more general noise on the eigenvalue oracle. 
In particular, treating the case of noise with bounded variance seems to be the most relevant for providing recovery guarantees that relate to practical experimental capabilities. 
It might also be possible to weaken our assumptions about the signal-to-noise ratio. 
A lower bound would help to clarify where the limits are to these types of algorithms. 

A further important open question is understanding the power of the structured circuits that we use for eigenvalue estimation. 
When using shallow depth Clifford circuits to prepare stabilizer bases for eigenvalue estimation, what recovery guarantees are possible? 
Is it still possible to efficiently reconstruct arbitrary sparse Pauli channels? 
Our heuristics suggest that pseudo-random and relatively shallow Clifford circuits allow sufficient randomness in the support that the algorithm can still have provable convergence guarantees, but it would be interesting to establish this rigorously. 

Finally, the most important open problem is to use our algorithms on real experiments to characterize noise, improve calibration of a device, or customize an error correction procedure.

\acknowledgements
This work was supported by the US Army Research Office grant numbers W911NF-14-1-0098 and W911NF-14-1-0103, and the Australian Research Council Centre of Excellence for Engineered Quantum Systems (EQUS) grant number CE170100009.

\appendix

\section{Walsh-Hadamard ordering}\label{app:WHTordering}
In this paper we use a variant of the Walsh-Hadamard transform where the ordering is determined by the commutation relations between the Paulis. 
The natural (bit-wise) ordering of a WHT matrix can be calculated from the tensor product as:
\begin{equation}
    \text{WHT (natural ordering) } = \begin{pmatrix}1&1\\1&-1\end{pmatrix}^{\otimes n}\,.
\end{equation}
In this case, like the sequency order and dyadic order variants of the WHT, we reorder the columns of the transform matrix. 
Unless otherwise expressly noted, we use a WHT where the $(i,j)$th entry of the Hadamard transform matrix is given by $(-1)^{\langle i,j\rangle}$, where the inner product is the symplectic inner product introduced above. 
The advantages of using this variant of the WHT is that when it is used to transform eigenvalues to error rates and vice-versa (\Cref{eq:WHT,eq:reverse:WHT}), the position of each Pauli in the transformed vector remains constant.

\bibliographystyle{apsrev4-2}
\bibliography{sparse}

\end{document}